%% file: main_v3_arxiv.tex
\documentclass[11pt]{article} 
\usepackage{setspace}
\usepackage{times}
\usepackage[top=1.0in, bottom=1.0in, left=1in, right=1in]{geometry}
\usepackage[usenames,dvipsnames]{color}
\usepackage{xcolor}
\usepackage[colorlinks=true,citecolor=blue,linkcolor=violet]{hyperref}
\usepackage{hyperref}
\usepackage{amsmath}
\usepackage{url}
\usepackage{bm} 
\usepackage{pdflscape}
\usepackage{amsbsy} 
\usepackage{tikz} 
\usepackage[nameinlink]{cleveref} 
\usepackage{mathtools}
\usepackage{tablefootnote} 
\usepackage{natbib}
\usepackage{amsfonts} 
\usepackage{amssymb}  
\usepackage{amsthm}   
\usepackage{wasysym} 
\usepackage{amsxtra}
\usepackage{bm}
\usepackage{enumitem}
\usepackage{multirow}
\usepackage{multicol}
\usepackage{caption, subcaption}
\usepackage{nomencl}
\usepackage{color,colortbl}
\usepackage{dsfont}
\usepackage{bold-extra}
\usepackage{fancybox}
\usepackage{fancyhdr}
\usepackage{graphics}
\usepackage{graphicx}
\usepackage{latexsym}
\usepackage{longtable}
\usepackage{wrapfig}
\usepackage{algorithmic}
\usepackage{algorithm2e}
\usepackage{float}
\usepackage{comment}

\usepackage{caption}
\newtheorem{theorem}{Theorem}

\newtheorem{proposition}[theorem]{Proposition}	
\newtheorem{corollary}[theorem]{Corollary}
\newtheorem{assumption}[theorem]{Assumption}	

\newtheorem{remark}{Remark}
\usepackage[normalem]{ulem}
\newcommand{\myref}[2]{\hyperref[#2]{#1 \ref*{#2}}} 
\usepackage{array}
\newcolumntype{H}{@{}>{\iffalse}c<{\fi}@{}} 

\begin{document}
\title{Statistical Validation of Contagion Centrality in Financial Networks}
\author{Agathe Sadeghi\thanks{Stevens Institute of Technology, School of Business, Hoboken, NJ 07030. \textit{asadeghi@stevens.edu} Corresponding author.} \and Zachary Feinstein\thanks{Stevens Institute of Technology, School of Business, Hoboken, NJ 07030.}}

\date{}
\maketitle
\abstract{
In this paper, we introduce an impact centrality measure to evaluate shock propagation on financial networks capturing a notion of contagion and systemic risk contributions, permitting comparisons of these risks over time. 
In addition, we provide a statistical validation method when the network is estimated from data, as is done in practice. This statistical test allows us to reliably assess the computed centrality values.
We validate our methodology on simulated data and conduct empirical case studies using financial data. We find that our proposed centrality measure increases significantly during times of financial distress and is able to provide insights into the (market implied) risk-levels of different firms and sectors.
}\\


\newpage

\section{Introduction}\label{sec:intro}
The growing complexity of financial systems requires the creation of a framework that can assess the interconnections among firms, see e.g.\ \cite[Figure 1]{feinstein2023interbank}. These connections play a crucial role in capturing systemic events and shocks that can lead to system instability. Researchers have shown growing interest in modeling financial system using network representations to gain a deeper understanding of these effects \citep{huang13, anderson19, jackson20, agosto20}. Emerging as a significant driver of instability, contagion, often described as the propagation of information or shocks across nodes, results in cascading effects that spread throughout the entire network, akin to an epidemic \citep{jalili17}.
Within this paper, we construct a network centrality measure designed to capture financial contagion. Due to the proposed construction, we are able to, e.g., statistically test our centrality measure; as far as the authors are aware, this is the first network centrality that can be analytically\footnote{For some centrality measures, confidence intervals can be calculated using bootstrapping which is computationally demanding, see e.g. \cite{anton20}.} tested, as opposed to other papers which test the connectivity of the network and then build the centrality measures \citep{billio12, gauber23, basu24}.

Evaluating a network involves two key aspects: the topology of the network and the determination of the importance of each node within it, referred to as the centrality measure. \cite{markose12} emphasizes that the network's topology plays a pivotal role in assessing its vulnerability.

Due to the recent financial events, it is important to view the financial system as a complex network, where firms serve as nodes and financial dependencies as links \citep{battiston}.
There are multiple perspectives on financial networks, e.g., (i) informational \citep{fink16}, (ii) interbank or default contagion \citep{acemoglu16}, (iii) portfolio overlap or price-mediated contagion \citep{cont17}, and (iv) cross-ownership \citep{fichtner17}. The construction of financial networks involves various methods, including correlation \citep{pacreau}, partial correlation \citep{kenett14,milington}, principal components analysis \citep{ballester16}, and analysis of financial dependencies based on balance sheets and contracts between institutions \citep{eisenberg1}. Studies have demonstrated that interconnectedness between institutions can be represented using weighted links, which indicate correlations between their portfolios \citep{diebold14, cabrales17} and that interpretable networks can be built on linear models \citep{diebold9,kumar22}.

Failures or contagion can be triggered by common factors, such as market declines, or by idiosyncratic shocks affecting individual institutions, like spread widening \citep{benoit17}. The consequences of a shock can be broken down into two parts: the immediate failure of an institution triggered by a sufficiently large shock, and the spread of the shock throughout the system due to network effects. While many studies focus on the direct collapse of a firm (fundamental default) resulting from an external shock, they often overlook the aggregated impact of network-induced contagion (default by contagion) as noted by \cite{cont} and \cite{deng21}. It is important to distinguish between the first-order effect, which could be the default of a firm or the occurrence of a shock impacting the firm, and the higher-order impacts due to these stresses. As highlighted by \cite{battiston13,veraart2020distress}, distress contagion can propagate through the network, causing widespread losses, even when there are no defaults within financial system. To measure these effects, it is essential to define a suitable centrality measure.

Centrality measures serve as valuable tools to quantitatively assess the structural significance of nodes within a network. As stated by \cite{darcangelis16}, a higher centrality score indicates that a particular property is more fitting for a given node. Different centrality measures offer distinct insights into various centrality dimensions \citep{hevey18}. In addition, the centrality notion becomes more crucial when the networks are not symmetric \citep{capponi15} which is the case in financial systems. When examining propagation within networks, \cite{pacreau} highlights that the centrality of a node indicates its sensitivity to market fluctuations and reveals its connectivity within the network. Nodes with high centrality are often connected to many other nodes or are neighbors to densely interconnected subnetworks. In related work, \cite{ghanbari18} examines the relationship between the number of failed nodes and centrality measures. Along this line, \cite{battiston} introduces DebtRank, utilizing feedback-centrality to measure the impact of a shock on a single node. Building on this concept, \cite{puliga14} presents Group Debtrank to quantify the centrality of nodes when small shocks affect firms and concludes that during crisis periods, the centrality measure tends to increase. \cite{bardos} extends the framework to generalized DebtRank, with applications explored in \cite{ferracci22} and \cite{carro22}. \cite{alexandre21} applies differential DebtRank to the Brazilian interbank market, while \cite{jiamin23} introduces GuaranteeRank to evaluate default risk in intercorporate credit guarantee networks. Furthermore, there are centrality indices meaning the index is a (mainly linear) combination of different centrality measures, see e.g. \cite{jaramillo14, wang21}. For a literature on centrality measures refer to \cite{bloch23}.

It is crucial to select an appropriate centrality measure that aligns with the attribute under consideration \citep{singh20}. In this work, we introduce a simple centrality measure that reflects the spread of a shock over the network structure based on the Leontief inverse \citep{leon1,leon2}. Widely used in input-output analysis in economics \citep{antras,antras2,antras3}, the Leontief inverse has also found applications in network science more recently \citep{,moran}. Certain studies in the financial network field have employed the Leontief inverse to establish the network structure directly \citep{elliot14,darcangelis16}, rather than applying it to compute the centrality.

Prior studies have typically produced a single value as the centrality measure. To ensure the validity of research findings, it becomes essential to gain a comprehensive understanding of the sensitivity of centrality measures \citep{borgatti6}. This need for accuracy analysis is further highlighted in the context of large and complex networks, where potential issues like missing data or the presence of hidden variables (confounders) may arise \citep{carley3}. The accuracy of centrality measures in network theory is influenced primarily by the deviation between the estimated network and the true underlying network. Existing literature on error analysis is limited by its focus on common centrality measures and tends to be case-specific, considering scenarios with predefined error types, levels, densities, and node numbers \citep{smith13,lee15,niu15}. Moreover, many of these measurement errors are qualitatively assessed \citep{martin19} and are primarily restricted to undirected and unweighted networks. Additionally, a common assumption made is that the true network is known, allowing for the introduction of various error types. However, this assumption is often unrealistic, as the true network is typically unknown in practice.

\subsection{Primary Contributions}
In this paper, we construct weighted directed networks, where nodes represent firms and the links depict interconnections between them. Herein we utilize the Katz-Bonacich (KB) centrality \citep{katz53,bonacich72}, referred to as KB, to quantify node centrality during shock scenarios, illustrating a node's position when hit by a shock. Notably, it considers the propagation and contagion risks inherent in the network topology, and boils down the shock propagation dynamics to a closed form. However, rather than solely utilizing this analytical form of this centrality measure, we add on a statistical test layer to determine if, e.g., the determined centrality is merely an artifact of the data. 

We define KB at pair-, node-, and system-levels, offering different granularities for analysis. Pair-level KB between two firms demonstrates the first firm's position when the second firm experiences a shock. Node-level KB of a firm is the total resulting shock when it is stressed. System-level KB aggregates multiple scenarios to portray overall centrality changes under equally probable situations. We also allow the association of weights with node- and system-level centralities.

Our approach involves statistical analysis to assess the robustness of the centrality measure. We derive the asymptotic distribution of KB using a Taylor expansion and formulate hypotheses to test its significance—specifically, whether KB significantly differs from zero and whether its value varies statistically across different nodes. Since centrality measures can be quite noisy, a statistically validated result could indicate no significant difference from zero, see e.g. \ \cite[Figure 14]{tolga14}. We validate our framework using simulated data, demonstrating the theoretical distribution's compatibility with empirical observations. In this simulated setup, the true value of KB consistently falls within the confidence intervals of the theoretical distribution at a 97.5\% confidence level.

Applicability in financial contexts is a key aspect of our paper. We evaluate the performance of our framework using two financial datasets -- Credit Default Swap Index (CDX) and equity tick data. In the CDX case study, KB effectively captures underlying market dynamics, outperforming both degree centrality and the leading eigenvalue. In the case of tick data, KB provides insights into the market and institutions' states but does not suffice for standalone investment decisions. However, a blend of node- and system-level KB dynamics along with institutional background information can guide investment strategies effectively.

The structure of this paper is as follows: \myref{Section}{sec:meas} outlines the approach used to construct the network and introduces the Katz-Bonacich centrality measure. With this definition, we detail the statistical framework for the KB in \myref{Section}{sec:err} and validate this framework using simulated data. We apply the KB approach to financial data scenarios in \myref{Section}{sec:case} to assess framework's applicability to capture stress events. Finally, in \myref{Section}{sec:conc}, we summarize the key findings of the paper and propose potential avenues for future research.
\section{Network Centrality}\label{sec:meas}

A network consists of a set of nodes and links. The connectivity pattern between nodes is captured by the $M \times M$ real-valued adjacency matrix $A\in \mathbb{R}^{M \times M}$, where $M$ denotes the number of nodes, and its elements $a_{ij}$, for $i,j=1,...,M$, indicate the presence or absence of connections between nodes. For our purpose, the elements $a_{ij}$ represent the connectivity weights between nodes $i$ and $j$ \citep{papana17}, specifically capturing the propagation of shocks between these nodes. In this paper, we consider the nodes to be the firms and the links to be the connections.
\begin{assumption}[Stationarity condition]\label{assumption:stationary}
For the remainder of this work, we assume that the leading eigenvalue of the adjacency matrix is strictly less than one.
\end{assumption}

We introduce a centrality measure that as opposed to eigenvector centrality (see, e.g., \cite{markose12-1,jorge,jose}) but similar to DebtRank \citep{debtrank}, is comparable over time. This measure defined as in \myref{Equation}{eq:leon} below, is readily interpretable in the context of financial contagion. 
\subsection{Katz-Bonacich Centrality Measure}
Our centrality measure is derived from Katz-Bonacich centrality (KB) \citep{katz53,bonacich72} defined as $KB_\alpha=(I-\alpha A)^{-1}\mathbf{1}$ where $\alpha > 0$ is a decay parameter, $I$ is the identity matrix and $\mathbf{1}$ is a vector of ones.
Notably, when $\alpha < \frac{1}{\lambda_m}$, where $\lambda_m$ denotes the leading eigenvalue of the adjacency matrix $A$, then the matrix inverse $(I - \alpha A)^{-1}$ is the Leontief inverse of $\alpha A$. 
Following \myref{Assumption}{assumption:stationary}, we can restrict $\alpha \in (0,1]$. If $\alpha \in (0,1)$ then this decay parameter assigns higher weight to shorter paths (direct relations). However, $\alpha = 1$ provides a natural decay to shock propagation through the network.\footnote{As the leading eigenvalue of $A$ is less than one, higher powers of $A$ contribute progressively less. This property ensures that indirect effects naturally decay in magnitude.}
For this reason, for the remainder of this work we will solely consider $\alpha = 1$ so as to consider the unadulterated network and denote $KB := KB_1 = (I - A)^{-1}\mathbf{1}$.

\begin{proposition}
Under \myref{Assumption}{assumption:stationary}, the pair-level KB is well-defined and can be simplified via the Leontief inverse:
\begin{align}
    \sum_{t = 1}^\infty A^t = (I-A)^{-1} - I\, ,
    \label{eq:leon}
    \tag{pair-level KB}
\end{align}
where $I$ is the $M \times M$ identity matrix.
\end{proposition}
\begin{proof}
$\sum_{t = 0}^\infty A^t = I + \sum_{t = 1}^\infty A^t = (I-A)^{-1}$ per \cite{leon1}.
\end{proof}
\begin{remark}
This methodology can handle networks with cycles and does not impose the constraint of acyclic graph structures, as seen in some prior research (see, e.g., \cite{teter14}).
\end{remark}
\begin{remark}
KB is not solely influenced by the strength of the links but also by the overall structure of the network. A densely connected network with comparatively weaker links may yield a higher centrality measure than a sparsely connected network with stronger links.
\end{remark}
To compute the weighted node-level centrality when node $i$ is shocked, we consider the weighted summation of each row of the centrality measure matrix: 
\begin{gather}
    \boldsymbol{c} = [\,(
    I-A)^{-1}-I\,]\,\boldsymbol{w}\,,
    \tag{node-level KB}
    \label{eq:sumcent}
\end{gather}
where $\boldsymbol{c}$ is the vector of centrality and $\boldsymbol{w}=(w_{1},...,w_{M})^\top$ are node weights. This weight vector can be customized based on the specific context of the research. For example, it can be defined as the ratio of liabilities to assets of the institutions or the market capitalizations of the firms. Incorporating weights allows us to assign varying degrees of importance to each node when assessing the system's centrality. As highlighted by \cite{ali2016}, conventional centrality measures may not fully capture the systemic importance of financial firms as the size and structure of their balance sheets play a crucial role in understanding the potential risks they pose to the entire system. Larger values in this centrality measure indicate a higher level of potential spillover, particularly during times of crisis.

Finally, we define the system-level centrality as the total sum of all node-level centralities $\boldsymbol{c}_i$, in line with the approach outlined by \cite{debtrank}. This value provides a useful proxy for the total systemic risk captured in financial data.
\begin{gather}
    \sum_{i=1}^M \boldsymbol{c}_i\,.
    \tag{system-level KB}
    \label{eq:sys_lcm}
\end{gather}
\subsection{Financial Intuition of KB}
To present the KB measure more intuitively and in shock scenarios, in \myref{Figure}{fig:netw_schem} we present a schematic network diagram along with its shock propagation dynamics, intuitively similar to Impulse Response literature in econometrics \citep{koop96, pesaran98}. Consider four time series, corresponding to four nodes denoted as $W,X,Y$ and $Z$, which are connected as illustrated in the figure on the left. The shock propagation pattern is shown in the figure on the right side, under the scenario in which node $W$ experiences shock $\epsilon$ at time $t=0$. Due to the connections with nodes $X$ and $Y$, this shock subsequently propagates and affects these nodes at time $t=1$. Following the network, the shock then cycles back to node $W$ and hits $Z$ at $t=2$. This cascading effect goes on until the shock dissipates.
\input{schematic_network}
In order to formulate the cumulative impact of a shock more generally, we assume $M$ time series and define an initial shock vector $\boldsymbol{\epsilon} \in \mathbb{R}^M$, where element $i$ represents the magnitude of the shock to node $i$. The shock propagation dynamics can be shown as:
\begin{equation*}
    \begin{gathered}
    \label{eq:shock_prop}
    \boldsymbol{s}_t = \boldsymbol{s}_{t-1} A, \quad \boldsymbol{s}_0 = \boldsymbol{\epsilon} \quad \Rightarrow \quad \boldsymbol{s}_t = \boldsymbol{\epsilon} A^t,
    \end{gathered}
\end{equation*}
where $\boldsymbol{s}_t$ is the state of the system at time $t$ and $A$ is the adjacency matrix. By applying the adjacency matrix $A$ to each state iteratively, we can trace the propagation of the shock over time. Specifically, $\boldsymbol{s}_1=\boldsymbol{s}_0A$, $\boldsymbol{s}_2=\boldsymbol{s}_1A=\boldsymbol{s}_0A^2$, and so on. Each term in the series $\boldsymbol{\epsilon}A^t$ represents the effect of the shock propagating through the network over different time intervals, accounting for higher-order interactions and contagion effects. The summation
\[ \sum_{t=1}^\infty \boldsymbol{s}_t = \boldsymbol{\epsilon}A+\boldsymbol{\epsilon}A^2+... = 
\boldsymbol{\epsilon}\,\underbrace{(A + A^2 + ...)}_{\text{pair-level KB}}\] 
captures the cumulative effect of the shock propagation and represents the centrality of the system. 
Each element in this matrix represents the pair-level KB, i.e., the centrality of node $j$ (corresponding to the column index) when node $i$ (corresponding to the row index) is hit by a shock. Note that this pair-level KB is independent of the initial shock $\boldsymbol{\epsilon}$.
\begin{remark}
Though we constructed the above with a static network, we note that both the adjacency matrix and the weight vector for the node-level KB can be time-varying. This will be used in the financial case studies of \myref{Section}{sec:case}. For the purposes of this work, recall that the adjacency matrix is provides weighted propagation of shocks through the network. Notably, this adjacency matrix can be derived using various methodologies, e.g., correlations, partial correlations, a vector autoregression model, or linear causality networks.
\end{remark}
\begin{remark}
When aiming to isolate idiosyncratic components from time series, one approach is to account for the influence of common factors. Techniques like PCA or Singular Value Decomposition (SVD) or spectral reduction \citep{ricciardi22} can identify dominant structures associated with these factors. However, to preserve the uniqueness of individual series, it is more appropriate to filter directly using observed factors relevant to the data, such as those derived from economic or financial theory.
\end{remark}
\section{Statistical Analysis}\label{sec:err}
We extend the scope of traditional centrality calculations by developing the asymptotic distribution of the KB measure. This advancement not only allows for the computation of confidence intervals, but also facilitates statistical hypothesis testing. By deriving the distributional properties of the KB, we provide a framework to evaluate the significance of contagion risk estimates. This contributes to a deeper understanding of the measure's reliability and makes it possible to compare results across different networks or over time with a higher degree of statistical rigor.

\subsection{Distribution of KB}
In practice, a discrepancy between the true adjacency matrix $A$ and its estimated counterpart $\hat{A}$ potentially exists. We denote this difference as $\Delta$, such that $\hat{A} = A + \Delta$. For $T$ data points, taking advantage of the Taylor series expansion, we can express the KB in terms of the discrepancy as follows:
\begin{gather}
    (I-\hat{A})^{-1} = (I-A)^{-1} + (I-A)^{-1}\Delta(I-A)^{-1} + O\,(\frac{1}{T})\,.
    \label{eq:leon_taylor}
\end{gather}

By applying the Delta method \citep{delta_econ}, we are able to deduce the asymptotic distribution for the pair-level KB.
Throughout this section we consider the shorthand notation where $\cdot j$ denotes the $j$th column and $i\cdot$ denotes the $i$th row, respectively. 
\begin{theorem}[KB Distribution] 
Assume $\Delta_{\cdot j} \sim N(0,\Sigma_{jj})$ where $Cov(\Delta_{\cdot i},\Delta_{\cdot j}) = \Sigma_{ij}$.
Then for every $i,j$ we have the following asymptotic distribution for the pair-level KB:\label{theo:lcm_dist}
\begin{align*}
&\sqrt{T}\left[(I-\hat{A})^{-1}_{ij} - (I-A)^{-1}_{ij}\right]\\
&\qquad \sim N\left(0,\left[(I-A)^{-1}_{\cdot j}\right]^\top \left[(I-A)^{-1}_{i\cdot}\Sigma_{lk}\left[(I-A)^{-1}_{i\cdot}\right]^\top\right]_{lk}(I-A)^{-1}_{\cdot j}\right)\, ,
\end{align*}
\end{theorem}

The detailed proof of \myref{Theorem}{theo:lcm_dist} is given in the \myref{Appendix}{app:pair-level}. This consideration of the pair-level KB distribution permits us to consider, also, the distribution of node-level KB.
\begin{corollary}\label{cor:node-level_lcm}
Consider the same setting as in \myref{Theorem}{theo:lcm_dist}. The node-level KB, defined in \myref{Equation}{eq:sumcent}, asymptotically follows the normal distribution:
\begin{align*}
&\sqrt{T}\left[(I-\hat{A})^{-1}_{i\cdot}\boldsymbol{w} - (I-A)^{-1}_{i\cdot}\boldsymbol{w}\right]\\
&\qquad \sim N\left(0,\left[(I-A)^{-1}\boldsymbol{w}\right]^\top \left[(I-A)^{-1}_{i\cdot}\Sigma_{lk}\left[(I-A)^{-1}_{i\cdot}\right]^\top\right]_{lk}(I-A)^{-1}\boldsymbol{w}\right)\,
\end{align*}
for every $i = 1,...,M$
where the weight vector is denoted as $\boldsymbol{w} = (w_{1},...,w_{M})^\top$. 
\end{corollary}

The proof of \myref{Corollary}{cor:node-level_lcm} is given in \myref{Appendix}{app:node-level}.
\subsection{Hypothesis Testing}
In this section, we aim to explore a potential application of the node-level KB distribution as outlined in \myref{Corollary}{cor:node-level_lcm}. Specifically, we propose two statistical tests on the KB centrality to determine its reliability in practice.
First, we introduce a statistical test to determine whether the node~$i$ KB significantly exceeds zero, i.e., the node contributes significantly to the propagation of a shock within the network. A second hypothesis test examines whether the centrality values of two nodes differ significantly from each other.
\subsubsection{Different from Zero}\label{sec:diff_zero}
In this hypothesis test, we aim to determine whether the node-level centrality significantly differs from zero. This approach adds an additional layer of validation to the centrality measure, ensuring that the results are not merely driven by noise but are statistically meaningful.
\begin{equation}
\begin{gathered}
    H_0: (I-A)^{-1}_{i\cdot}\boldsymbol{w}-w_{i} = 0\, ,\\
    H_1: (I-A)^{-1}_{i\cdot}\boldsymbol{w}-w_{i} > 0\, .
    \label{eq:null_hyp}
\end{gathered}
\end{equation}

We note that when working with networks, constructing them effectively requires a substantial number of data points. Therefore, the asymptotic distribution of the network is sufficient. 
\begin{corollary}[Test Statistic and Distribution] 
Consider the setting of \myref{Theorem}{theo:lcm_dist} with the test statistic:
\begin{gather*}
Z := \frac{\sqrt{T}\left[(I-\hat{A})^{-1}_{i\cdot}\boldsymbol{w} - w_i\right]}{\left(\left[(I-\hat{A})^{-1}\boldsymbol{w}\right]^\top \left[(I-\hat{A})^{-1}_{i\cdot}\Sigma_{lk}\left[(I-\hat{A})^{-1}_{i\cdot}\right]^\top\right]_{lk}(I-\hat{A})^{-1}\boldsymbol{w}\right)^{\frac{1}{2}}} \,.
\end{gather*}
Under the null hypothesis, $Z \sim N(0,1)$ asymptotically.
\label{theo:t_dist}
\end{corollary}
This result directly follows from \myref{Corollary}{cor:node-level_lcm}, the Central Limit Theorem and Slutsky's Theorem. The formal proof is given in \myref{Appendix}{app:test-stat}.
\subsubsection{Different from Each Other}
To compare different shock scenarios, specifically evaluating the contagion risk when node~$i$ is shocked versus when node~$j$ is shocked, the following hypothesis test can be carried out:
\begin{equation}
\begin{gathered}
    H_0: (I-A)^{-1}_{i\cdot}\boldsymbol{w}-(I-A)^{-1}_{j\cdot}\boldsymbol{w} = 0\, ,\\
    H_1: (I-A)^{-1}_{i\cdot}\boldsymbol{w}-(I-A)^{-1}_{j\cdot}\boldsymbol{w} \neq 0\, .
    \label{eq:null_hyp_node}
\end{gathered}
\end{equation}
 \begin{corollary}[Test Statistic and Distribution] 
Consider the setting of \myref{Theorem}{theo:lcm_dist} with the test statistic:
\begin{gather*}
Z := \frac{\sqrt{T}\left[(I-\hat{A})^{-1}_{i\cdot}\boldsymbol{w} - (I-\hat{A})^{-1}_{j\cdot}\boldsymbol{w}\right]}{\left(Var\left[(I-\hat{A})^{-1}_{i\cdot}\boldsymbol{w}\right]+Var\left[(I-\hat{A})^{-1}_{j\cdot}\boldsymbol{w}\right]-Cov\left[(I-\hat{A})^{-1}_{i\cdot}\boldsymbol{w},(I-\hat{A})^{-1}_{j\cdot}\boldsymbol{w}\right]\right)^{\frac{1}{2}}}\, 
\end{gather*}
where the variances can be obtained from \myref{Corollary}{cor:node-level_lcm} and the covariance structure is given by:
\begin{align*}
Cov\left[(I-\hat{A})^{-1}_{i\cdot}\boldsymbol{w},(I-\hat{A})^{-1}_{j\cdot}\boldsymbol{w}\right]=\left[(I-\hat{A})^{-1}\boldsymbol{w}\right]^\top \left[(I-\hat{A})^{-1}_{i\cdot}\Sigma_{lk}\left[(I-\hat{A})^{-1}_{j\cdot}\right]^\top\right]_{lk}(I-\hat{A})^{-1}\boldsymbol{w}\,.
\end{align*}
Under the null hypothesis, $Z \sim N(0,1)$ asymptotically.
\label{theo:t_dist_diff}
\end{corollary}
The proof follows identically to that of \myref{Corollary}{theo:t_dist} and is given in \myref{Appendix}{app:test-stat}.
\subsection{Numerical Validation}\label{sec:synth}
To evaluate the validity of the KB, i.e., alignment between the theoretical and empirical distributions and their true counterparts, we conduct a case study using simulated data. We generate a dataset of 3 linearly dependent time series with iid noise terms and length of 600. The empirical analysis involves 10,000 simulations to obtain a robust representation. Subsequently, we calculate the mean and variance of the theoretical and empirical distributions of the KB centrality measure and compare these distributions with the true values of the centrality.

\myref{Figure}{fig:dist_leon_theo_emp_tr} presents the node-level KB distributions. A close examination of the figures reveals a strong alignment between all three sets of values. The true centrality value falls within the 95\% confidence interval calculated from the theoretical distribution. 
We note that the theoretical distribution may exhibit a slight bias due to the theoretical distribution's mean being estimated from one realization of the adjacency matrix (see \myref{Figure}{fig:node2}).

\begin{figure}[!tb]
\begin{minipage}{\linewidth}
    \begin{subfigure}{0.33\linewidth}
    \centerline{\includegraphics[width=\linewidth]{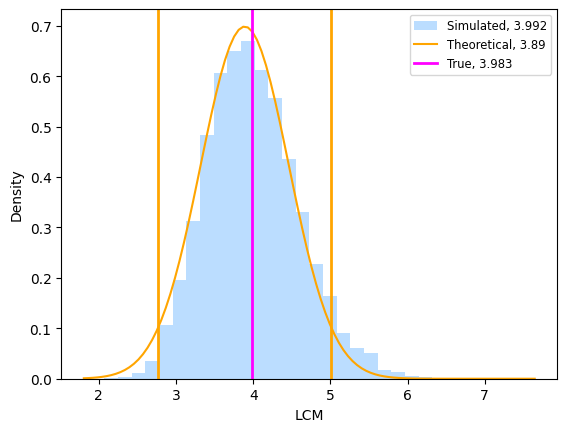}}
    \caption{\footnotesize{node 1}}
    \end{subfigure}\hfill
    \begin{subfigure}{0.33\linewidth}
    \centerline{\includegraphics[width=\linewidth]{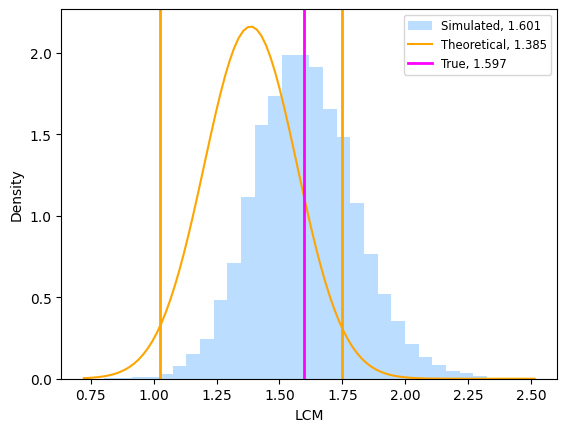}}
    \caption{\footnotesize{node 2}}
    \label{fig:node2}
    \end{subfigure}
    \begin{subfigure}{0.33\linewidth}
    \centerline{\includegraphics[width=\linewidth]{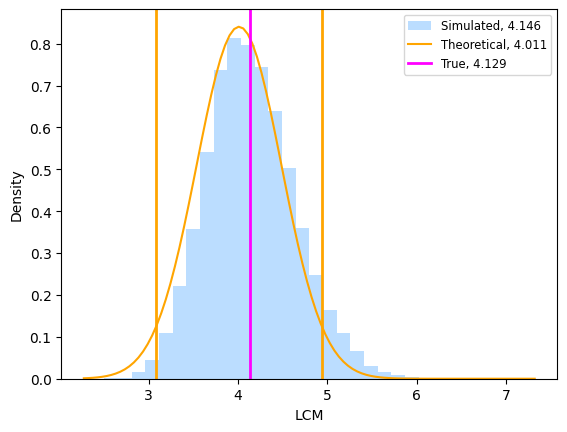}}
    \caption{\footnotesize{node 3}}
    \end{subfigure}
    \caption{\footnotesize{Distribution of node-level KB.}}\label{fig:dist_leon_theo_emp_tr}
\end{minipage}
\end{figure}
The QQ plots in \myref{Figure}{fig:qq_leon_theo_emp_tr} provide a visual comparison between the theoretical and empirical distributions. The points in the plots are close to the red regressed line, indicating similarity between the two distributions. However, slight deviations can be observed at the tails of the distributions. This is due to our use of the asymptotic theoretical distribution. Consequently, some bias is expected when finite data is considered, e.g., from the errors of order $1/T$ in the Taylor expansion of \myref{Equation}{eq:leon_taylor}.
\begin{figure}[!tb]
    \begin{subfigure}{0.33\linewidth}
    \centerline{\includegraphics[width=\linewidth]{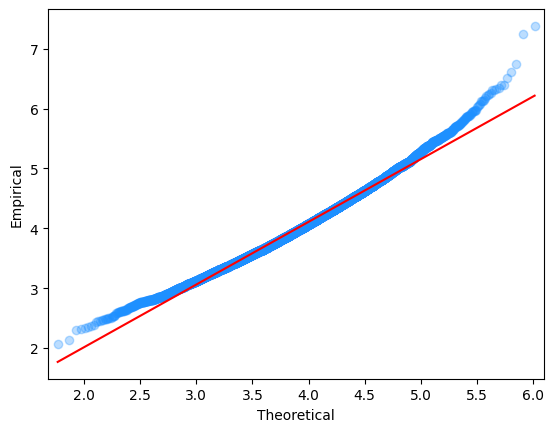}}
    \caption{\footnotesize{node 1}}
    \end{subfigure}\hfill
    \begin{subfigure}{0.33\linewidth}
    \centerline{\includegraphics[width=\linewidth]{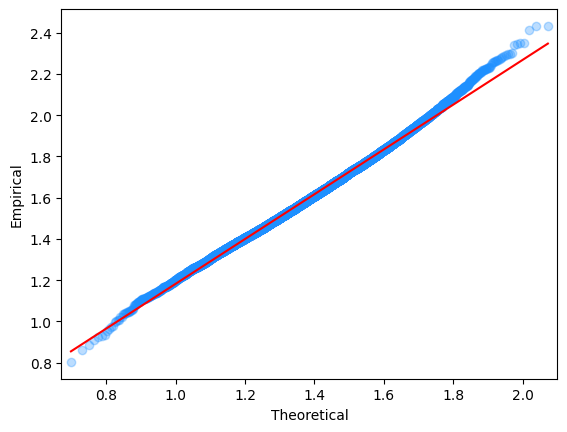}}
    \caption{\footnotesize{node 2}}
    \end{subfigure}
    \begin{subfigure}{0.33\linewidth}
    \centerline{\includegraphics[width=\linewidth]{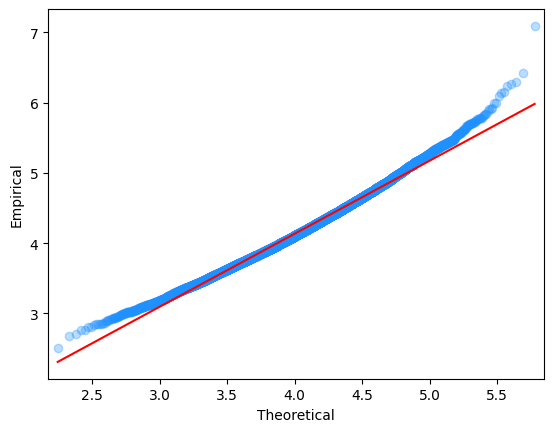}}
    \caption{\footnotesize{node 3}}
    \end{subfigure}
    \caption{\footnotesize{QQ plot of node-level KB comparing the theoretical and empirical distribution}.}\label{fig:qq_leon_theo_emp_tr}
\end{figure}
\section{Financial Case Studies}\label{sec:case}
In this section we present two case studies in order to demonstrate the financial implications of our framework. In the first case study, we utilize CDX spread data to construct networks, from which we analyze the KB for a spanning history of over 11 years. In the second one, we leverage high-frequency equity tick data to capture the information in the stock market. We highlight the years with turmoils to demonstrate how the KB reflects market behavior. For simplicity, we obtain the adjacency matrix from a Vector Autoregressive Linear System considering one lag (VAR(1)). For more details refer to \myref{Appendix}{app:linear_sys}.

\begin{remark}
Herein, to overcome parameter instability of the adjacency matrix when calibrated from a vector autoregressive model, we apply moving average smoothing to all following results. We refer the interested reader to, e.g., \cite{diebold9, kilian18, kumar22, diebold23} on difficulties with the use of vector autoregressive models.
\end{remark}
\subsection{Credit Default Swap Index Market Contagion}\label{sec:cds_data_new}
In most of the existing literature, financial networks are typically constructed based on contracts between institutions, particularly in the banking sector, where a core of large banks are connected to peripheral ones. Numerous studies have explored this core-periphery network structure, mostly in the context of interbank lending \citep{upper11,battiston15,bardoscia,bichuch18,blasque18,farbodi17}. However, only a limited number of studies have delved into derivative markets, specifically Credit Default Swaps (CDS) \citep{brun13,balsques14}.

Our dataset is obtained from Bloomberg, breaking down Markit's Investment Grade (IG) and High Yield (HY) CDX spread data\footnote{For more information on these indices, we refer the interested reader to \myref{Appendix}{app:case}.} into sectors based on the Global Industry Classification Standard (GICS)\footnote{The breakdown is at level 1 and excludes real estate.}. This breakdown yields 10 sectors for each grade, totaling 20 time series. These indices have been reporting since January 1, 2011. Our analysis spans from this starting date to August 31, 2023. We use a one-year time window, sliding it monthly, resulting in 146 time windows. This approach allows us to assess the contagion risk across different periods, enabling a deeper understanding of the derivative market's network's dynamics over time.

To make the time series stationary and reduce variability, we compute the daily log return of the time series 
$r_{i,t}=log(\frac{s_{i,t+1}}{s_{i,t}})$
where $s_{i,t}$ is the CDX spread of the sector $i$ at time $t$. Using the framework with 1-day lag, we build our VAR(1) network from these daily returns for each time window.
Afterwards, we calculate the KB for the system. For simplicity, we perform the analysis and statistical tests of \myref{Section}{sec:diff_zero} based on the unweighted centrality setting ($\boldsymbol{w} = \boldsymbol{1}$).

Visualizing the financial networks for February and March 2020 in \myref{Figure}{fig:cds_netw_fin}, we observe that the network in March 2020 exhibits thicker links compared to the network in February 2020, indicating stronger connectivity. Moreover, the March network appears to be more densely connected. Focusing on the March network, we observe a core-periphery like structure. The core nodes, characterized by more and thicker links either incoming or outgoing, feature the Consumer Discretionary and Energy indices. A more detailed analysis on what this entails is given below.
\begin{figure}[!tb]
\centering
\begin{minipage}{.45\textwidth}
  \centering
  \includegraphics[width=1\textwidth, height=190px]{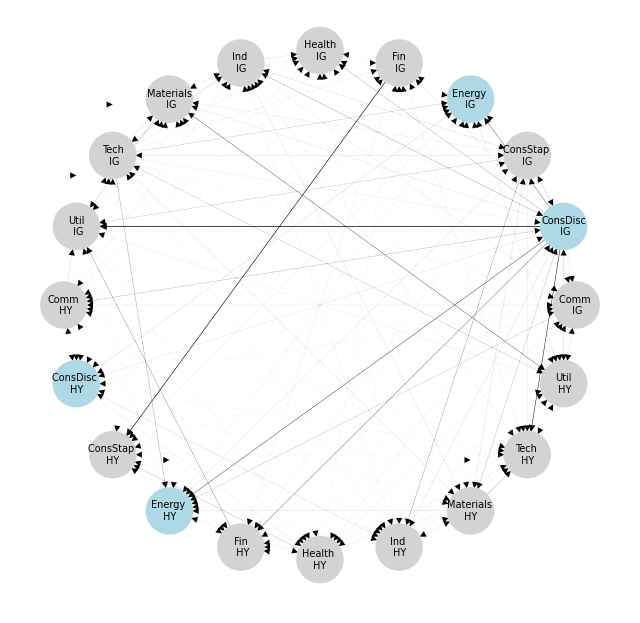}
\end{minipage}%
\begin{minipage}{.45\textwidth}
  \centering
  \includegraphics[width=1\textwidth, height=190px]{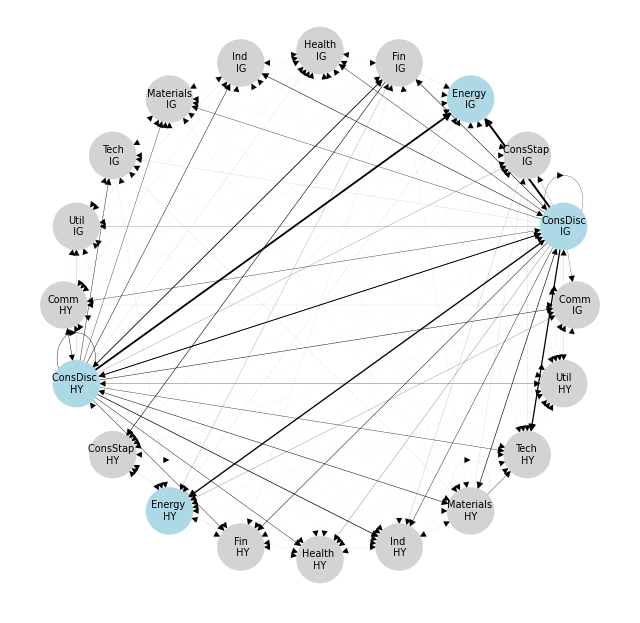}
\end{minipage}
\caption{\footnotesize{CDX network. \textbf{Left} February 2019-2020 \textbf{Right} March 2019-2020. The blue colored nodes are Energy and Consumer Discretionary, IG and HY}}
\label{fig:cds_netw_fin}
\end{figure}

In \myref{Figure}{fig:cdx_deg_eig}, we present the system-level KB centrality measure as a moving average across different time windows. The visualization indicates the overall level of contagion risk and sector instability, containing indices from both high yield and investment grade sectors. One key insight derived from the figure is the surge in KB during and after the Covid-19 pandemic, indicating a high contagion risk during this period \citep{yu21}. Comparing the time window from March 2019 to March 2020 with the preceding years, the measure obtains a value of about five times greater, reaching its peak in late 2020. Remarkably, when March is added to the time window, a sharp increase in KB is triggered instantly, while it gradually starts to drop close to the end of 2020. These dynamics are primarily influenced by several events related to the outbreak of Covid-19 in early 2020 leading to a widespread economic disruptions. The uncertainty triggered significant market volatility with concerns about the economic fallout, indicating a higher perceived risk. Additionally, there were vulnerabilities specific to sectors that were captured in the CDX data. Afterwards, the US government and Federal Reserve responded with monetary policy interventions and fiscal stimulus in order to support financial markets. Signs of recovery were shown also when vaccines were rolled out, as optimism grew regarding the potential economic rebound. Economic indicators improved and reduced the uncertainty about the long-term pandemic's impact. All these led to the gradual decline of risk and hence the KB measure.

We also notice a ramp up upon including mid-2022 data. This is mainly due to the increasing interest rates by Federal Reserve and some bankruptcy, merger and near-default events. It is worth mentioning that the elevated KB in late 2012 may be due to the lingering effects of the Great Recession and the 2011 European sovereign debt crisis. This analysis highlights the measure's robustness to economic events and interventions, showing its ability to capture the sector instability and contagion risk.

We plot the leading eigenvalue and thresholded degree centrality alongside KB in \myref{Figure}{fig:cdx_deg_eig} for the same adjacency matrix. We see that the leading eigenvalue exhibits some level of noise, displaying marginal increase during Covid-19 period. Towards the end of the date range, this value maintains a consistent trend without any notable elevation. Similarly, the degree centrality line appears rather stable throughout the figure. Though not displayed, we also calculated the betweenness and closeness centralities; these exhibited similar patterns to degree centrality. This suggests that leading eigenvalue and the other three centralities struggle to capture the contagion risk effectively.
\begin{figure}[!tb]
\centering
  \includegraphics[scale=0.6]{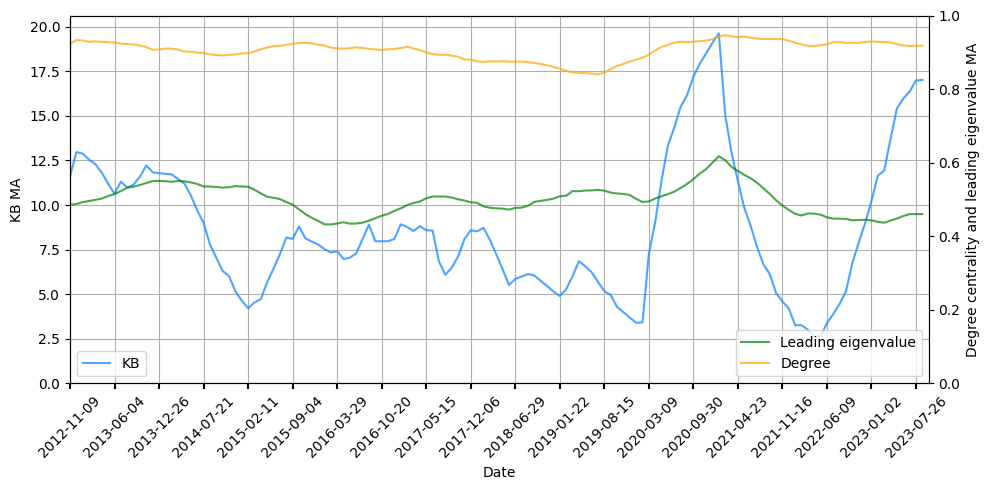}
\caption{\footnotesize{Moving average of CDX KB, degree centrality and leading eigenvalue}}
\label{fig:cdx_deg_eig}
\end{figure}

In \myref{Figure}{fig:cdx_lcm_dr}, we compare the KB to its unvalidated counterpart and the DebtRank measure \citep{debtrank} taken with respect to the same adjacency matrix. Unvalidated counterpart leaves out the hypothesis test in \myref{Equation}{eq:null_hyp}. DebtRank (DR) is a feedback centrality measure that distresses one node at a time and quantifies the propagation of systemic risk throughout the network. At each time step, nodes are classified as undistressed, distressed, or inactive.

From mid-2013 through the end of 2019, both the unvalidated KB and DR show a gradual rise, with DR beginning its increase earlier in 2013 and the unvalidated KB starting to climb in 2015. However, this period, from 2013 to 2019, was marked by relative financial stability,\footnote{\url{https://www.forbes.com/advisor/investing/bull-market-history/\#:~:text=Bull\%20Market\%20of\%202009\%2D2020\%3A\%20The\%20Great\%20Recession\%20Recovery}} with no significant systemic shocks, making the steady rise in both measures noticeable. However, this increase may reflect false alarms due to the lack of statistical validation. In contrast, the validated KB behaves differently. While some noise is present, there is no significant upward trend during this stable period. This highlights the importance of validation in filtering out spurious signals which allows the KB to reduce the likelihood of false positives.

During the Covid period, all three measures rise, capturing the financial turbulence. However, when examining DR, the increase in early 2021 is nearly identical in magnitude to the peak observed in September 2019, both reaching 0.8. The 2019 peak occurred without any clear signs of systemic instability, emphasizing the challenge of interpreting unvalidated metrics.

All three measures experience a decline after the end of 2021, followed by an upward trend as the financial system reacts to the risks introduced by rising interest rates. For more in-depth analysis, \myref{Table}{tab:cdx_deg_eig} shows the percentage changes in these measures during the two periods of financial stress.
\begin{figure}[!tb]
\centering
  \includegraphics[scale=0.6]{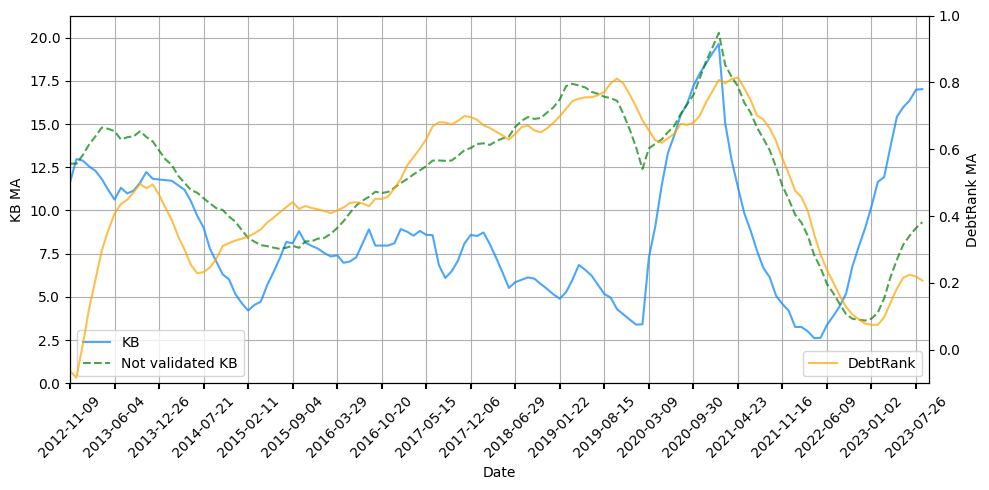}
\caption{\footnotesize{Moving average of CDX KB, not validated KB and DebtRank}}
\label{fig:cdx_lcm_dr}
\end{figure}
\begin{table}[!tb]
\footnotesize
\centering
\caption{\footnotesize{Relative change of centrality measures and leading eigenvalue during two periods of financial stress}}
\begin{tabular}{cccccc}
\hline\hline
Date & KB & Degree Centrality & Leading Eigenvalue & DebtRank & Unvalidated KB\\\hline
March 2019 - March 2020 & 5.327 & 0.123 & 0.025 & -0.387 & 0.763 \\
March 2022 - March 2023 & 7.260 & 0.003 & 0.077 & 3.773 & 8.533 \\ \hline
\end{tabular}
\label{tab:cdx_deg_eig}
\end{table}

\myref{Figure}{fig:cdx_consdisc} shows KB over time for the Consumer Discretionary sector. We display the unvalidated KB along with the 97.5\% confidence interval. When the confidence interval intersects the zero line, the validated KB is considered zero. In time window March 2019 to 2020, a notable increase in both investment grade and high yield CDX within the sector becomes evident. This time period coincides with the beginning and the outbreak of Covid-19 pandemic in U.S., disrupting economic activities across multiple sectors. Moreover, the lockdowns, social distancing and job losses affected the consumer behavior, leading to decreased spending on non-essential expenses such as leisure, travel and luxury items.

On top of this, the global spread of Covid-19, which had already affected China, disrupted supply chains worldwide. Disruptions in the production and distribution of goods created challenges for companies within the discretionary sector, hindering their access to raw materials and logistics. Consequently, many businesses in non-essential retail sector either closed or experienced a major decline in sales. Hence, due to the pandemic, consumer behavior shifted, supply chain systems had disruptions and economic uncertainties escalated. Combination of these factors resulted into a notable increase in the contagion risk and vulnerability in the U.S. consumer discretionary sector. For a comprehensive view of sector-specific figures of CDX KB, we refer the interested reader to \myref{Figure}{fig:cdx_sector} in the Appendix.
\begin{figure}[!tb]
\centering
\begin{subfigure}{0.45\textwidth}
  \includegraphics[width=\linewidth]{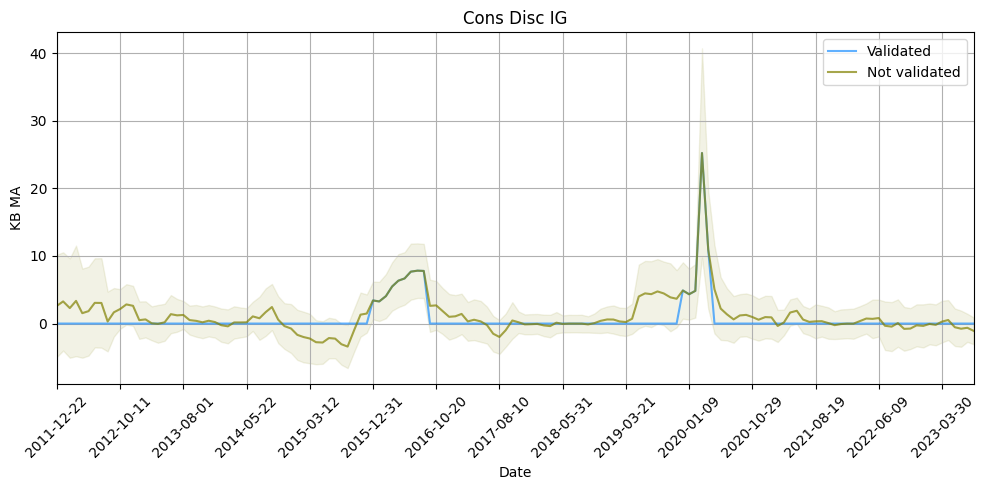}
\end{subfigure}\hfil
\begin{subfigure}{0.45\textwidth}
  \includegraphics[width=\linewidth]{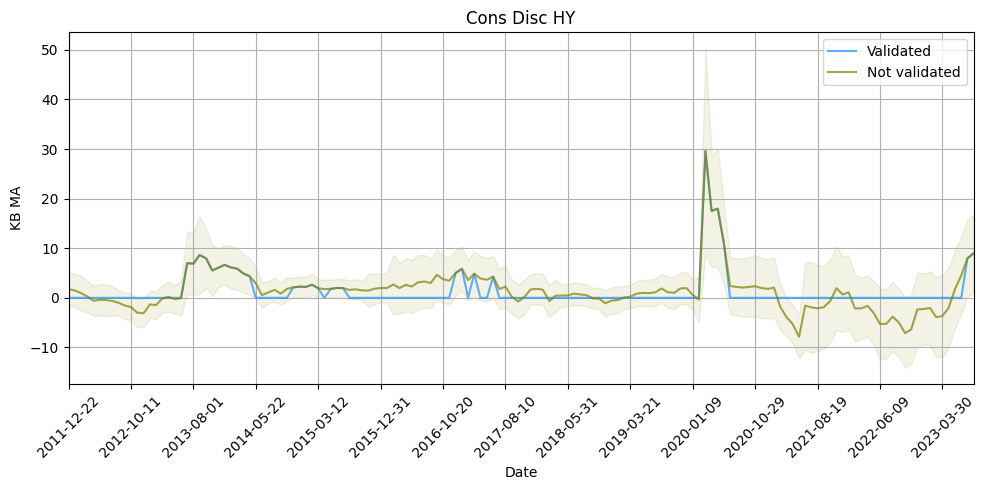}
\end{subfigure}
\caption{\footnotesize{Consumer Discretionary sector CDX KB \textbf{Left} Investment grade \textbf{Right} High yield}}
\label{fig:cdx_consdisc}
\end{figure}
%
\subsection{Equity Tick Market Contagion}\label{sec:tick_data}
The choice of data frequency is a critical consideration in network analysis. To ensure that dependencies are captured in the model, we must strike the right balance between a frequency that is not too low or too high. Given that stock prices swiftly adjust to market changes, especially to similar stock price fluctuations, and trading occurs at high frequency to capitalize on these changes, we opt for tick data. We use tick log return data for 20 Financials sector companies (provided in \myref{Table}{tab:cds_data} of \myref{Appendix}{sec:tick_data_table}) as well as both Lehman Brothers (LEH) and Bear Stearns (BSC). The tick data is obtained from Refinitiv, collected at one-second intervals, and spans from January 1, 2005 to December 31, 2010 aligning with \cite{jaramillo14}.

In \myref{Figure}{fig:tick_lcm}, we focus on the period surrounding the 2008 Financial Crisis to gain a deeper understanding of the centrality measure's behavior in the years preceding, during, and after the turmoil. We consider two versions of the system-level KB: (i) for the network including LEH and BSC, for which we run the analysis for as long as the data is available (October 8, 2008 indicated by a red dashed line in the figure) and (ii) for the network excluding these two firms. In both versions, we observe an upward trend starting from early 2007 with the peak KB reached in August 2008. A gradual decline is then observed when excluding LEH and BSC. This signifies that the contagion risk increased until reaching the peak of turmoil, subsequently decreasing after government interventions and as the economy and markets slowly returned to normalcy. This system-level trend closely aligns with the timeline of the financial crisis: the acquisition of Bear Sterns by JPMorgan on March 16, 2008, and the subsequent default of Lehman Brothers on September 15, 2008. The Wall Street Bailout law was passed on October 3, 2008, which included the \$700 billion Troubled Asset Relief Program (TARP) aimed at rescuing banks from defaulting.

We plot the node-level KB for three specific firms (American International Group (AIG), JPMorgan (JPM) and Goldman Sachs (GS)) alongside Lehman Brothers and Bear Stearns in \myref{Figure}{fig:tick_lcm}. We note that the node-level KB for LEH and BSC have similar behavior to the system-level KB until May 2008. Afterwards, they start to decline, due to which we experience a slight decline in the system-level KB as well. The node-level drop is aligned with recognition of the substantial risk these two companies carried by investors. Due to this market behavior, investors treated these companies idiosyncratically thus reducing the spillover effects to the rest of the system from shocks they might experience. This is why, despite the systemic contagion risk rising, the KB for LEH and BSC move inversely. This detachment is also reflected in the network structure. \myref{Figure}{fig:LEH_BSC_netw} shows the interdependencies between LEH and BSC and all other firms for mid March and September 2008 time windows. As is evident, the connectivity of these two companies diminishes in both quantity and significance during their corresponding drop in their KB.

Analyzing node-level KBs reveals some similarities and differences between the three tickers, AIG, JPM and GS, and the delisted ones. We notice that AIG, prominent in the insurance sector, exhibits a similar behavior to LEH and BSC; its KB took a downward turn following the bankruptcy and merger of LEH and BSC respectively. During the crisis, \$182 billion was injected into AIG to prevent its default by the Treasury and Federal Reserve Bank of New York, as it was believed that its collapse would jeopardize the integrity of other firms. Essentially, AIG came under government ownership and was safeguarded, which is reflected in its KB as it reaches nearly zero by early 2009. On the other hand, JPM and GS, for which the KB consistently increases even after October 8, 2008, diverge from the system-level KB. As these two companies were considered safer by investors, they maintained investments despite the systemic turmoil. However, under a shock scenarios to these companies, there would have been a high contagion risk. Looking into the stock prices of these tickers, we observe a significant decline in AIG's price, whereas JPM and GS experience a decrease but recover shortly afterwards. These patterns are aligned with the fluctuations seen in the KB, indicating its ability to mirror the dynamics of the market.
\begin{figure}[!tb]
\centering
   \centering
   \includegraphics[scale=0.6]{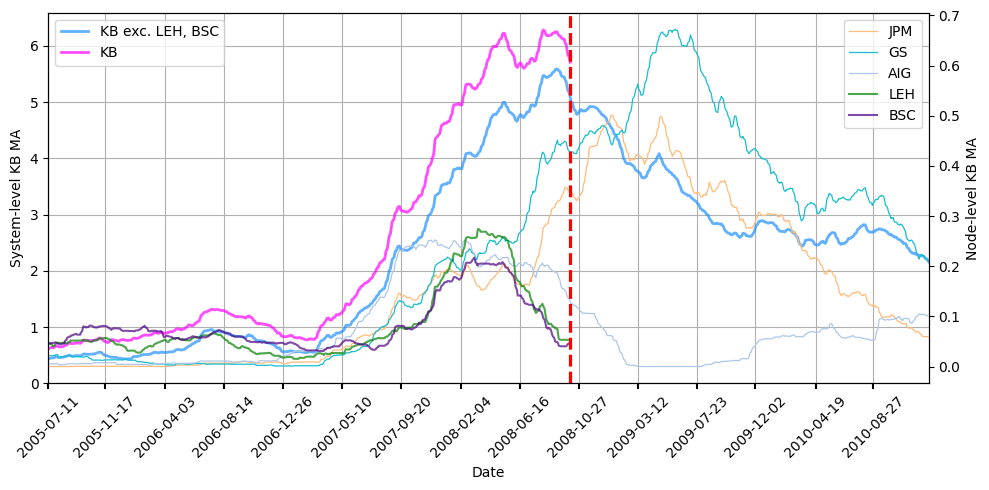}
    \caption{\footnotesize{Node- and system-level KB}}
\label{fig:tick_lcm}
\end{figure}
\begin{figure}[!tb]
\centering
\begin{minipage}{.45\textwidth}
  \centering
  \includegraphics[width=1\textwidth, height=170px]{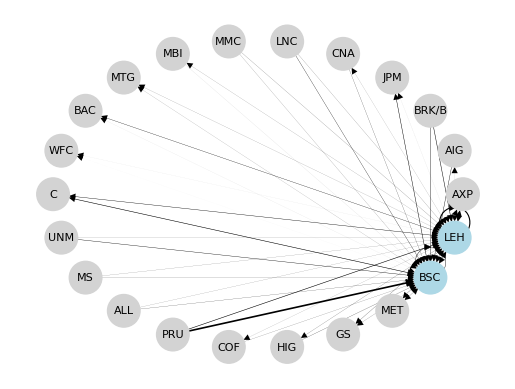}
\end{minipage}%
\begin{minipage}{.45\textwidth}
  \centering
  \includegraphics[width=1\textwidth, height=170px]{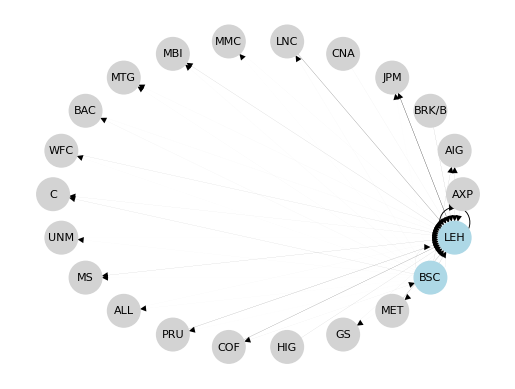}
\end{minipage}
\caption{\footnotesize{Incoming and outgoing edges of LEH and BSC network. \textbf{Left} Mid March 2008 \textbf{Right} Mid September 2008. The blue colored nodes are LEH and BSC}}
\label{fig:LEH_BSC_netw}
\end{figure}

\myref{Table}{tab:sys_node_dynamics} provides a concise summary of the logic above. When both node- and system-level KBs are on the rise, more instability is introduced coming from individual institutions. Conversely, a declining trend at the node-level while the system-level KB is increasing demands a closer look at the specific firm. Such divergence implies the firm is treated independently from the network and can imply different conditions: if the firm is struggling, portfolio detachment is triggered; if the firm is in a good shape, it is considered a safe investment opportunity. Similarly, when the system-level KB is decreasing while the node-level is increasing, it holds implications based on the firm's condition; a struggling firm means a high likelihood of instability, while a healthy one becomes densely interconnected within the network, drawing heavy reliance from investors. This interconnectedness elevates the risk of contagion under a shock scenario to these robust entities.
\begin{table}[!b]
\footnotesize
\centering
\caption{\footnotesize{Node- and system-level KB dynamics}}
\begin{tabular}{ccc}
\hline\hline
System-level & Node-level & Intuition \\\hline
$\uparrow$ & $\uparrow$ & More instability and risk in system (can result in a crisis)\\
$\uparrow$ & $\downarrow$ & Idiosyncratic treatment for the firm from investors \\
$\downarrow$ & $\uparrow$ & Heavy reliance on the firm \\
$\downarrow$ & $\downarrow$ & More stability and less risk in system (safe) \\
\hline
\end{tabular}
\label{tab:sys_node_dynamics}
\end{table}
\section{Conclusion}\label{sec:conc}
This paper presents a framework for constructing a weighted directed network, computing the centrality measure and statistically validating it. This offers a robust approach to assess contagion risk and vulnerability within financial systems. The KB centrality measure serves as a reliable proxy for contagion risk, allowing for meaningful comparisons across different time periods. To further enhance analysis of the network, we derive the distribution of the centrality measure, enabling the computation of confidence intervals and hypothesis testing. We show the alignment of the theoretical and empirical centrality distributions with the true centrality measure using simulated data. Leveraging real world financial data, CDX spread and equity tick data, we construct financial networks and evaluate the centrality of the system. During distressed years, we observe higher KB values, underscoring the measure's ability to capture systemic risk within the network and assess the investors' beliefs. KB can be used to form clusters of firms based on their susceptibility to contagion risk, thereby being useful for factor investing strategies \citep{peckman24} and risk management practices. Central institutions through this measure can be identified to determine which firms require more rigorous oversight.

Future research can broaden the scope of this work by exploring alternative network construction approaches, experimenting with different numbers of lags or markets \citep{mishra22}, and incorporating these variations into the network's structure. Such advancements could lead to the development of a multilayered centrality metric. From a financial perspective, applying this enhanced centrality measure to calculate risk metrics, such as Value at Risk, Expected Shortfall, CoVaR \citep{adrian16}, and MES \citep{acharya12}, would be particularly valuable. This extension could offer a more comprehensive understanding of the complex dynamics and interdependencies observed in financial data.
\section*{Disclosure of interest}
The authors have no conflicting interests to declare.
\section*{Declaration of funding }
No funding was received for this work.
\bibliographystyle{apalike}
\addcontentsline{toc}{chapter}{References}
{\footnotesize
\bibliography{bibb}}
\newpage
\appendix
\section{Proofs}\label{sec:app}
\subsection{Proof of Pair-Level KB Distribution (\myref{Theorem}{theo:lcm_dist})}\label{app:pair-level}
As we are focused on the asymptotic distribution of $(I-\hat{A})^{-1}_{ij}$, we consider the distribution associated with the first-order Taylor expansion
\[(I-\hat{A})^{-1}_{ij} = (I-A)^{-1}_{ij} + (I-A)^{-1}_{i\cdot}\Delta(I-A)^{-1}_{\cdot j}.\]
Under the assumptions of this theorem, the noise of the model $\Delta$ in \myref{Equation}{eq:a_hat} follows a multivariate normal distribution. In fact, due to the linear relation between the (first-order approximation of the) pair-level KB $(I-\hat{A})^{-1}_{ij}$ and the noise $\Delta$, it trivially follows that the pair-level KB must also follow a normal distribution. Thus, to complete this proof we only need to consider the mean and variance of this linear form.

First, we note that the mean of the noise $E[\Delta] = 0$ by assumption. Therefore, for any $T > 0\,$,
\[\sqrt{T} E\left[(I-\hat{A})^{-1}_{ij} - (I-A)^{-1}_{ij}\right] = 0\,.\]

Second, consider the variance of the noise. By the first-order Taylor expansion specified above,
\begin{align*}
var&\left(\sqrt{T}[(I-\hat{A})^{-1}_{ij} - (I-A)^{-1}_{ij}]\right) = T\,var\left((I-A)^{-1}_{i\cdot}\Delta(I-A)^{-1}_{\cdot j}\right)\\
&= T\left[(I-A)^{-1}_{\cdot j}\right]^\top var\left((I-A)^{-1}_{i\cdot}\Delta\right)(I-A)^{-1}_{\cdot j} \\
&= \left[(I-A)^{-1}_{\cdot j}\right]^\top \left[(I-A)^{-1}_{i\cdot}\Sigma_{lk}^T\left[(I-A)^{-1}_{i\cdot}\right]^\top\right]_{lk}(I-A)^{-1}_{\cdot j}\,,
\end{align*}
where $\Sigma_{lk}^T$ provides the covariance between $\Delta_{\cdot l}$ and $\Delta_{\cdot k}$ (with multiplier $\sqrt{T}$).
\subsection{Proof of Node-Level KB Distribution (\myref{Corollary}{cor:node-level_lcm})}\label{app:node-level}

This result follows trivial as the distribution of a linear combination of a multivariate normal is itself normally distributed. Following directly from \myref{Theorem}{theo:lcm_dist}, asymptotically, we find:
\begin{gather*}
\sqrt{T}E\left[(I-\hat{A})^{-1}_{i\cdot}\boldsymbol{w}-(I-A)^{-1}_{i\cdot}\boldsymbol{w}\right]=0\,,\\
T\,var\left((I-\hat{A})^{-1}_{i\cdot}\boldsymbol{w}-(I-A)^{-1}_{i\cdot}\boldsymbol{w}\right)=\left[(I-A)^{-1}\boldsymbol{w}\right]^\top \left[(I-A)^{-1}_{i\cdot}\Sigma_{lk}\left[(I-A)^{-1}_{i\cdot}\right]^\top\right]_{lk}(I-A)^{-1}\boldsymbol{w}\,.
\end{gather*}
\subsection{Proof of Test Statistic's Distribution (\myref{Corollary}{theo:t_dist}, \myref{Corollary}{theo:t_dist_diff})} \label{app:test-stat}
First, taking advantage of the weak law of large numbers, the estimated adjacency matrix $\hat{A}$ converges in probability to the true adjacency matrix $A$. By the continuous mapping theorem, the estimated Leontief inverse $(I-\hat{A})^{-1}$ converges in probability to the true Leontief inverse $(I-A)^{-1}$. Therefore, based on the Central Limit Theorem and Slutsky's Theorem, the test statistic $Z$ asymptotically follows a standard normal distribution. In addition, we note that this convergence does not introduce any bias due to the asymptotics used herein.
\section{Vector Autoregressive System (\myref{Section}{sec:case})} \label{app:linear_sys}
As we calibrate the adjacency matrix as the OLS estimates, we recall this construction for a generic multivariate model.
Given two matrices $X, Y \in \mathbb{R}^{T\times M}$ (over real numbers), where $T$ represents the number of time steps and $M$ denotes the number of variables, and an iid noise term $\mathcal{E}\in \mathbb{R}^{T\times M}$, we can adopt a linear regression framework as below:
\begin{equation*}
    \begin{gathered}
    Y=XA+\mathcal{E}\, ,\\
    \label{eq:a_hat}\hat{A}=(X^\top X)^{-1}X^\top Y\, ,
    \end{gathered}
\end{equation*}
where $A$ is the true coefficient matrix that governs the relationships between the elements of $X$ and $Y$. On the other hand, $\hat{A}$ represents the estimated coefficient matrix, which is obtained from the regression analysis. Define $Q := \lim_{T \to \infty} T^{-1} X^\top X$ (with limit taken in probability as the number of data points grows to infinity) and set $\Sigma_{lk} := \rho_{lk}\sigma_l\sigma_k Q^{-1}$ for each $l,k$ where $\rho_{lk}$ denotes the correlation coefficient between time series $l,k$ and $\sigma_k$ denotes the standard deviation of time series $k$ residuals.
\section{Additional Details and Figures for \myref{Section}{sec:cds_data_new}} \label{app:case}
\paragraph{Data} CDX spread data, which are available as Markit's North American Investment Grade and High Yield indices, closely reflect the credit quality and direction of the underlying basket in one tradable instrument. These indices are domiciled in North America, have 5 year tenor, USD denomination, quarterly coupon frequency, without restructuring, and semi-annual roll periods in March and September. The CDX IG is composed of 125 equally weighted credit default swaps on investment grade entities traded in the CDS market. Meanwhile, the CDX HY comprises 100 non-investment grade liquid entities within the CDS market.
\begin{figure}[htb]
    \centering
\begin{subfigure}{0.45\textwidth}
  \includegraphics[width=\linewidth]{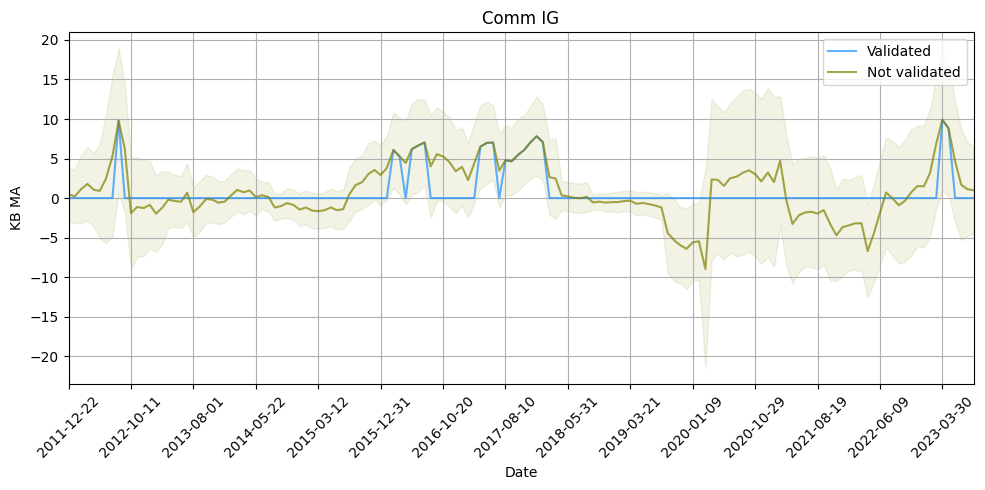}
  \caption{\footnotesize{Communication Services IG}}
\end{subfigure}\hfil
\begin{subfigure}{0.45\textwidth}
  \includegraphics[width=\linewidth]{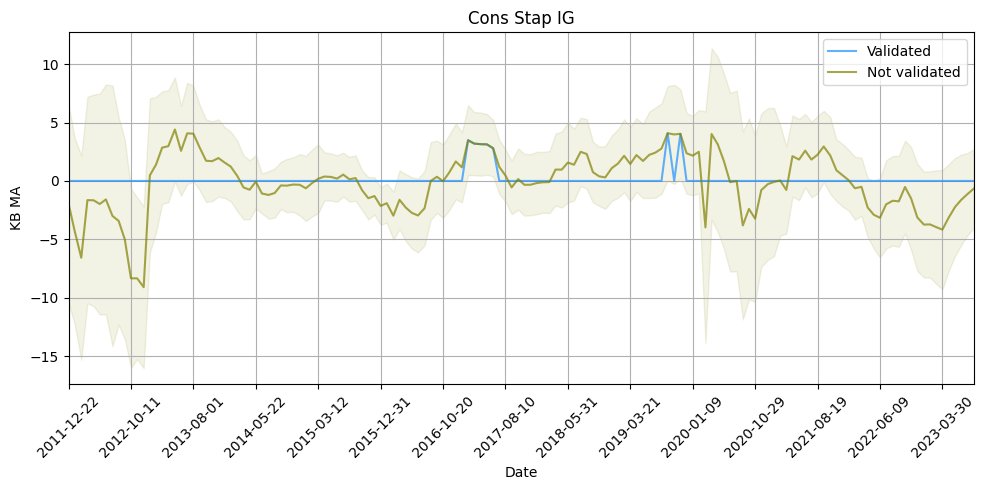}
  \caption{\footnotesize{Communication Services HY}}
\end{subfigure}

\begin{subfigure}{0.45\textwidth}
  \includegraphics[width=\linewidth]{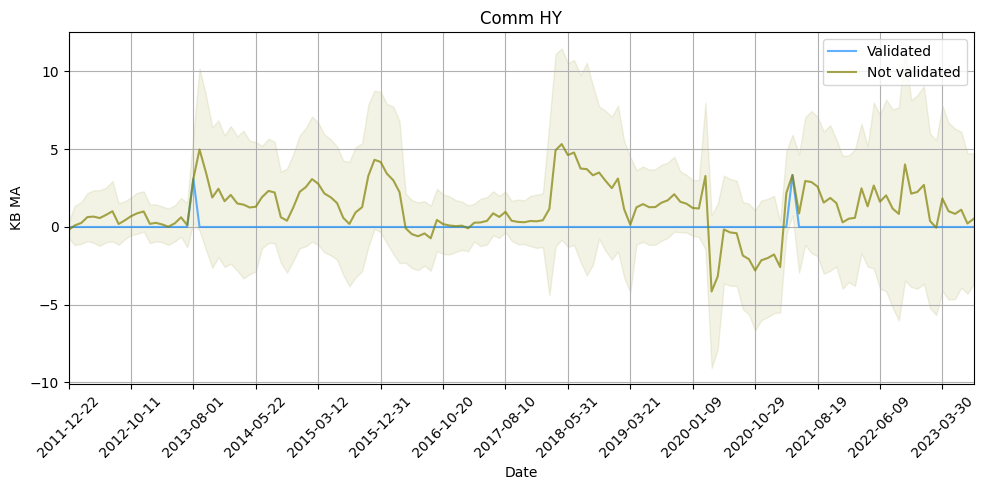}
  \caption{\footnotesize{Consumer Staples IG}}
  \label{fig:4}
\end{subfigure}\hfil
\begin{subfigure}{0.45\textwidth}
  \includegraphics[width=\linewidth]{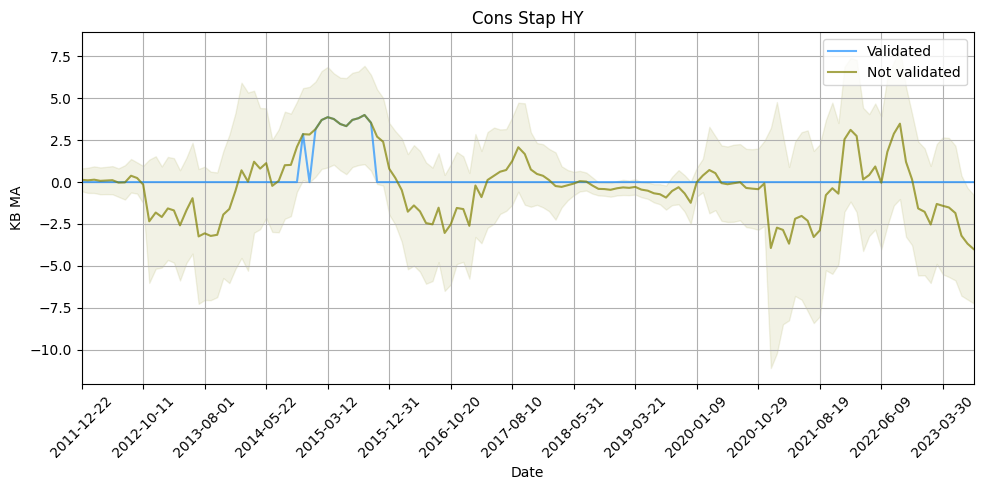}
  \caption{\footnotesize{Consumer Staples HY}}
  \label{fig:5}
\end{subfigure}\hfil

\begin{subfigure}{0.45\textwidth}
  \includegraphics[width=\linewidth]{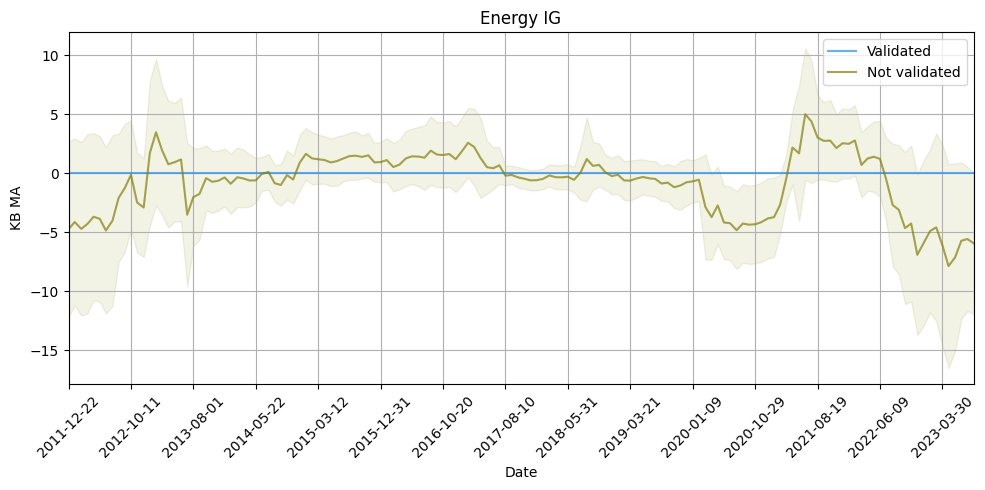}
  \caption{\footnotesize{Energy IG}}
  \label{fig:1}
\end{subfigure}\hfil
\begin{subfigure}{0.45\textwidth}
  \includegraphics[width=\linewidth]{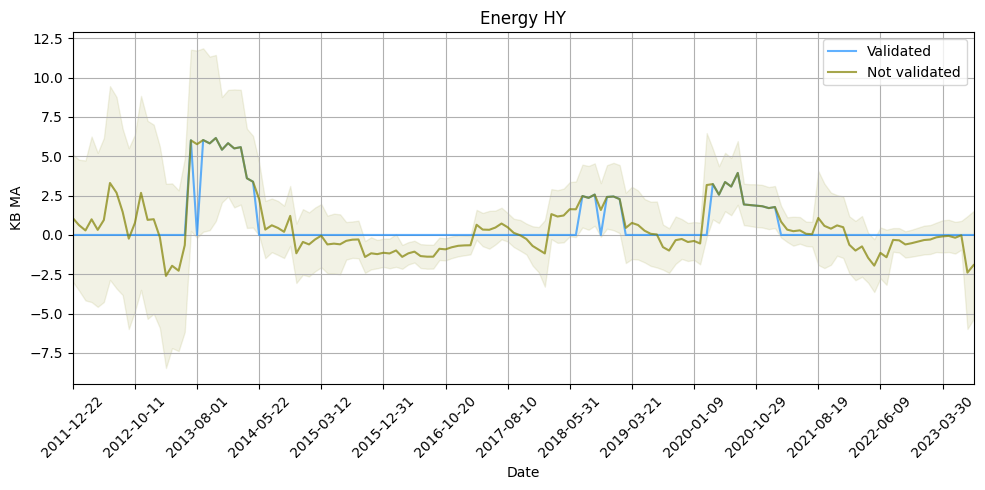}
  \caption{\footnotesize{Energy HY}}
  \label{fig:2}
\end{subfigure}

\begin{subfigure}{0.45\textwidth}
  \includegraphics[width=\linewidth]{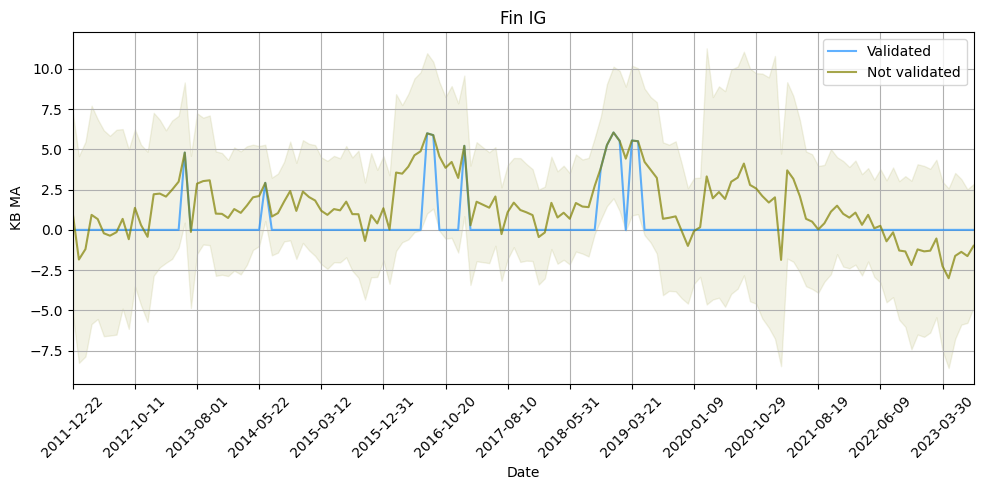}
  \caption{\footnotesize{Financials IG}}
  \label{fig:1}
\end{subfigure}\hfil
\begin{subfigure}{0.45\textwidth}
  \includegraphics[width=\linewidth]{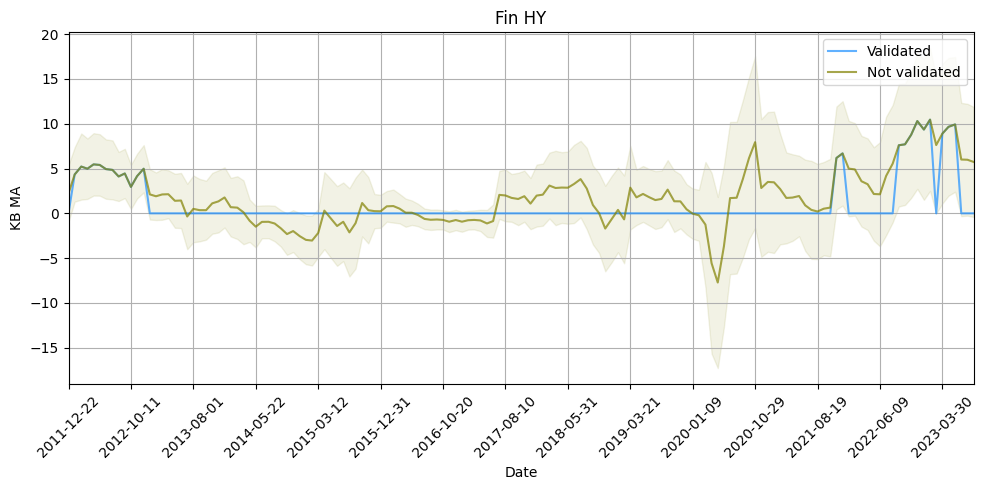}
  \caption{\footnotesize{Financials HY}}
  \label{fig:2}
\end{subfigure}

\begin{subfigure}{0.45\textwidth}
  \includegraphics[width=\linewidth]{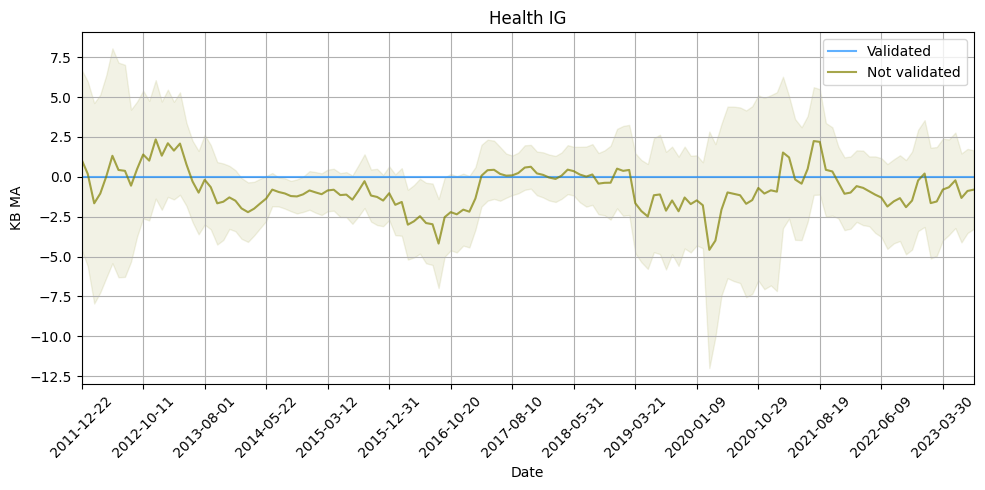}
  \caption{\footnotesize{Healthcare IG}}
  \label{fig:1}
\end{subfigure}\hfil
\begin{subfigure}{0.45\textwidth}
  \includegraphics[width=\linewidth]{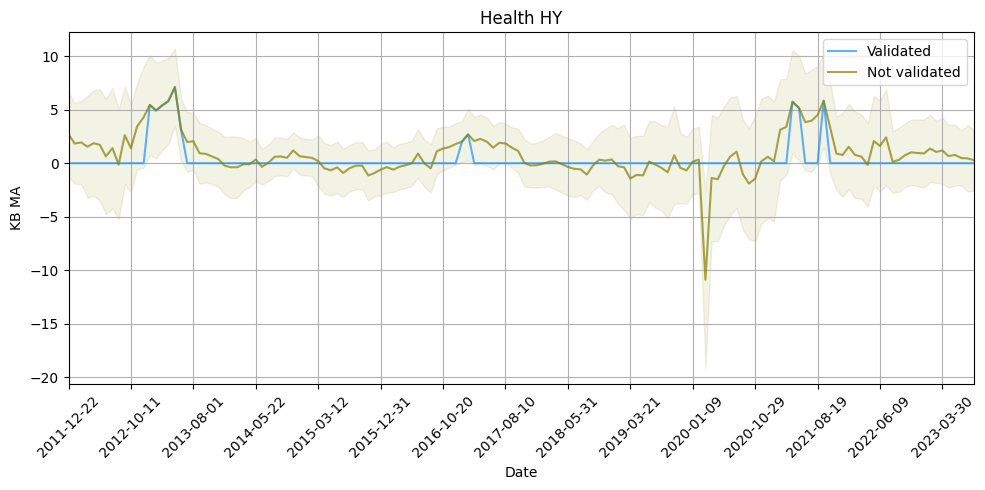}
  \caption{\footnotesize{Healthcare HY}}
\end{subfigure}
\label{fig:images}
\end{figure}

\begin{figure}[htb]\ContinuedFloat
    \centering
\begin{subfigure}{0.45\textwidth}
  \includegraphics[width=\linewidth]{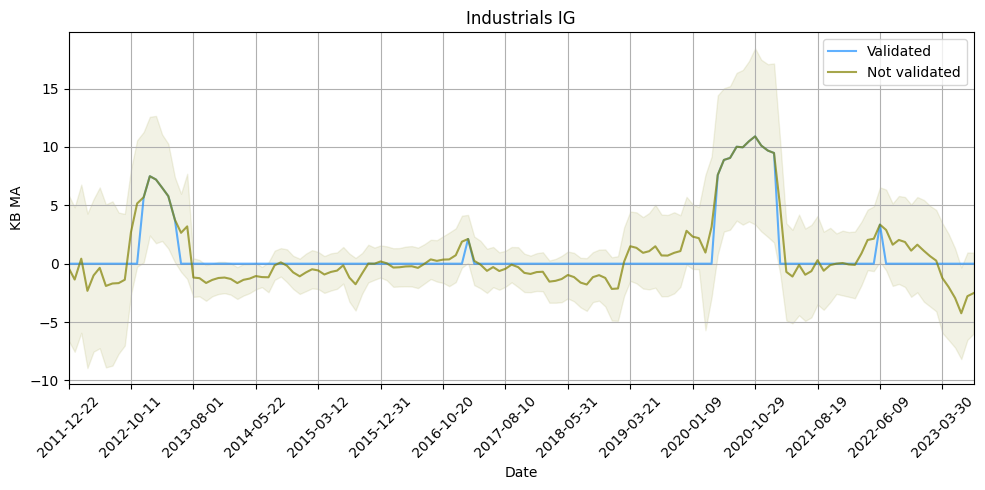}
  \caption{\footnotesize{Industrials IG}}
  \label{fig:1}
\end{subfigure}\hfil
\begin{subfigure}{0.45\textwidth}
  \includegraphics[width=\linewidth]{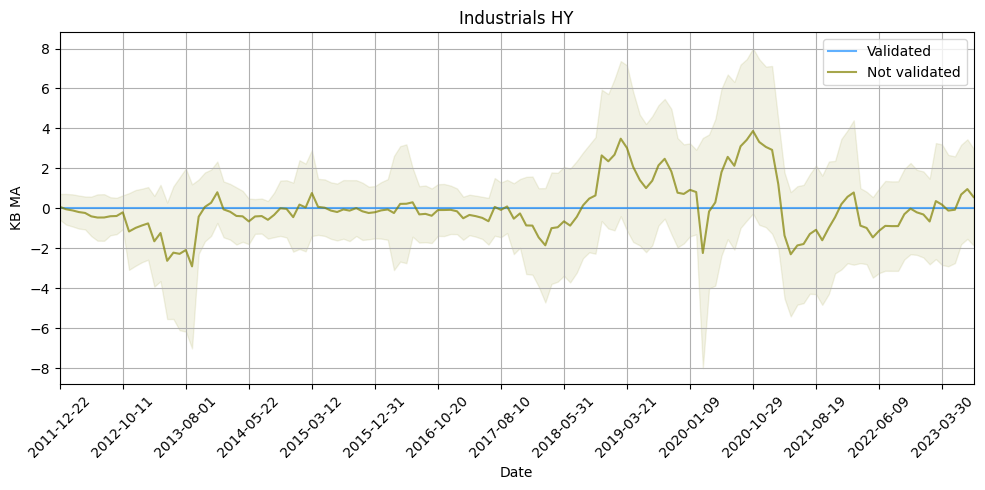}
  \caption{\footnotesize{Industrials HY}}
  \label{fig:2}
\end{subfigure}

\begin{subfigure}{0.45\textwidth}
  \includegraphics[width=\linewidth]{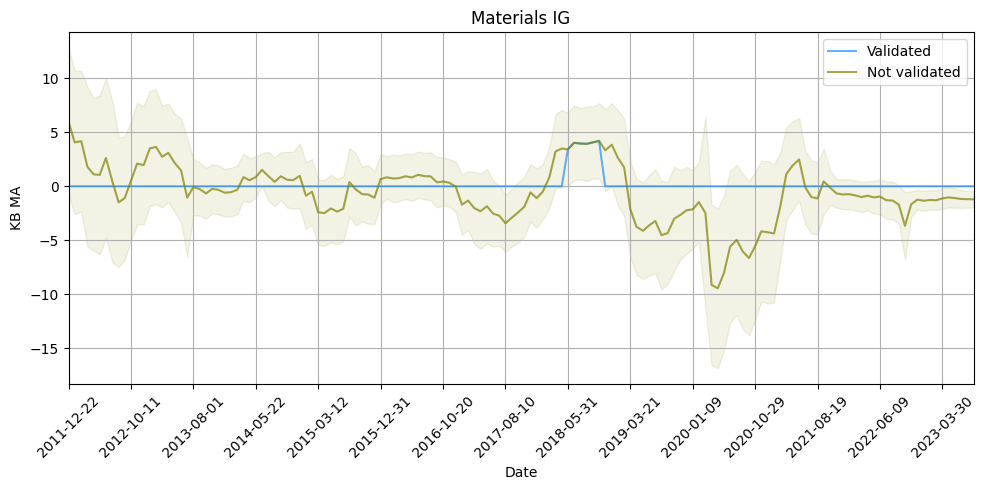}
  \caption{\footnotesize{Materials IG}}
  \label{fig:1}
\end{subfigure}\hfil
\begin{subfigure}{0.45\textwidth}
  \includegraphics[width=\linewidth]{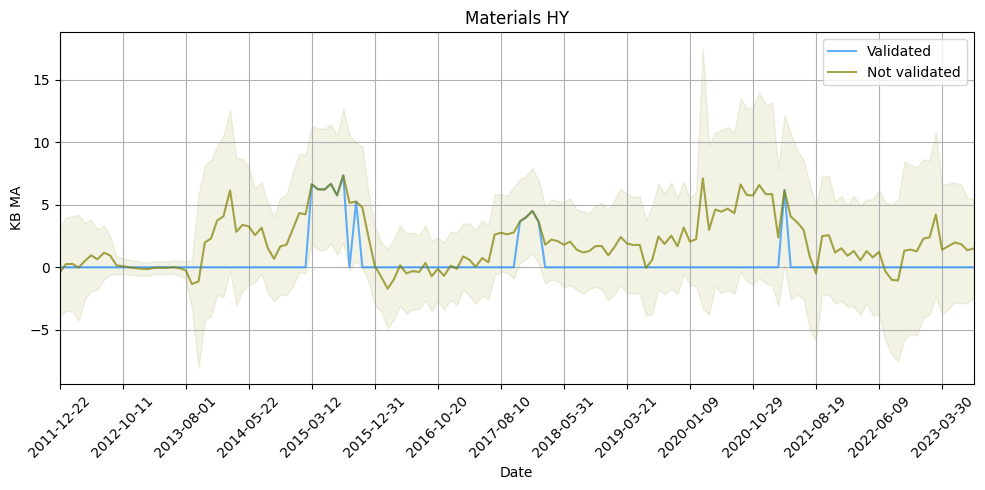}
  \caption{\footnotesize{Materials HY}}
  \label{fig:2}
\end{subfigure}

\begin{subfigure}{0.45\textwidth}
  \includegraphics[width=\linewidth]{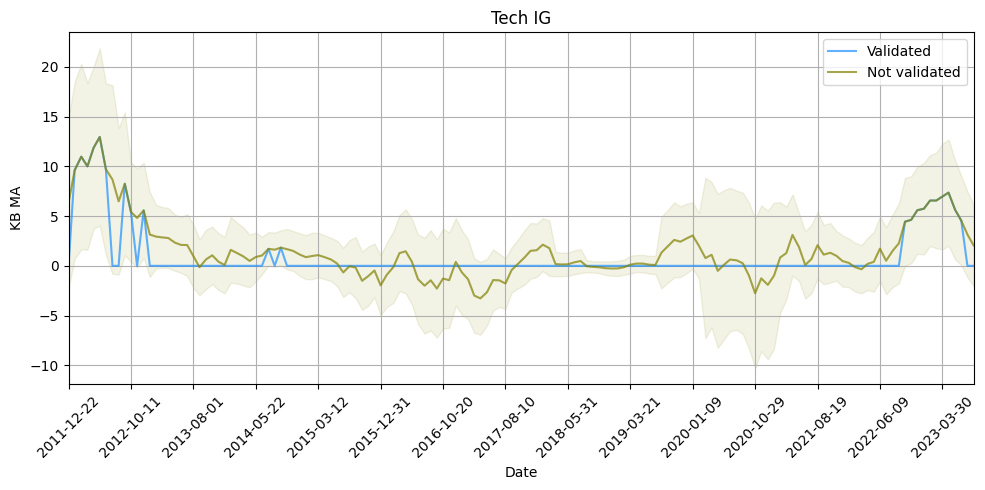}
  \caption{\footnotesize{Technology IG}}
  \label{fig:1}
\end{subfigure}\hfil
\begin{subfigure}{0.45\textwidth}
  \includegraphics[width=\linewidth]{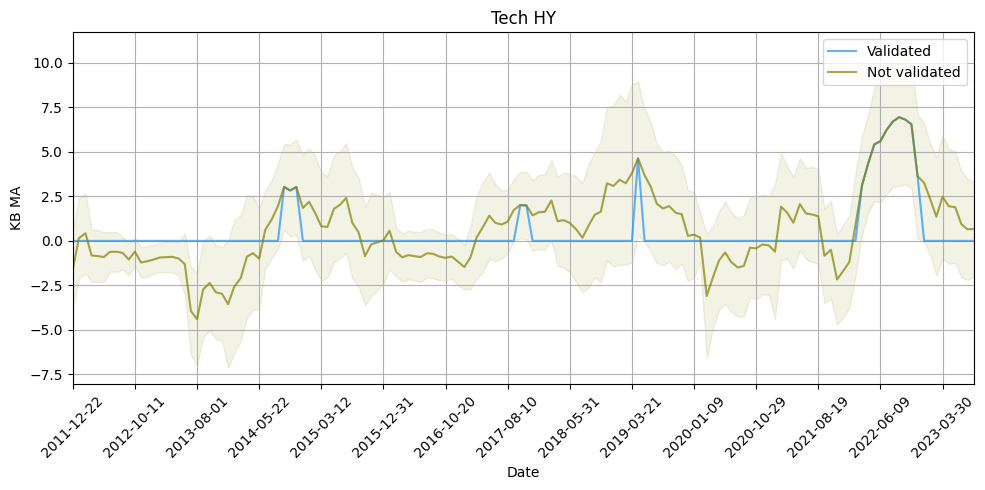}
  \caption{\footnotesize{Technology HY}}
  \label{fig:2}
\end{subfigure}

\begin{subfigure}{0.45\textwidth}
  \includegraphics[width=\linewidth]{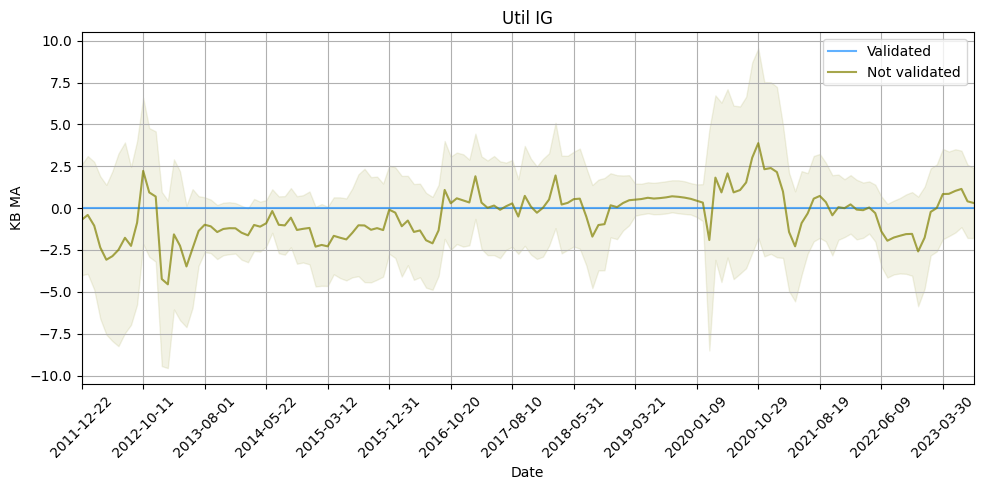}
  \caption{\footnotesize{Utilities IG}}
  \label{fig:1}
\end{subfigure}\hfil
\begin{subfigure}{0.45\textwidth}
  \includegraphics[width=\linewidth]{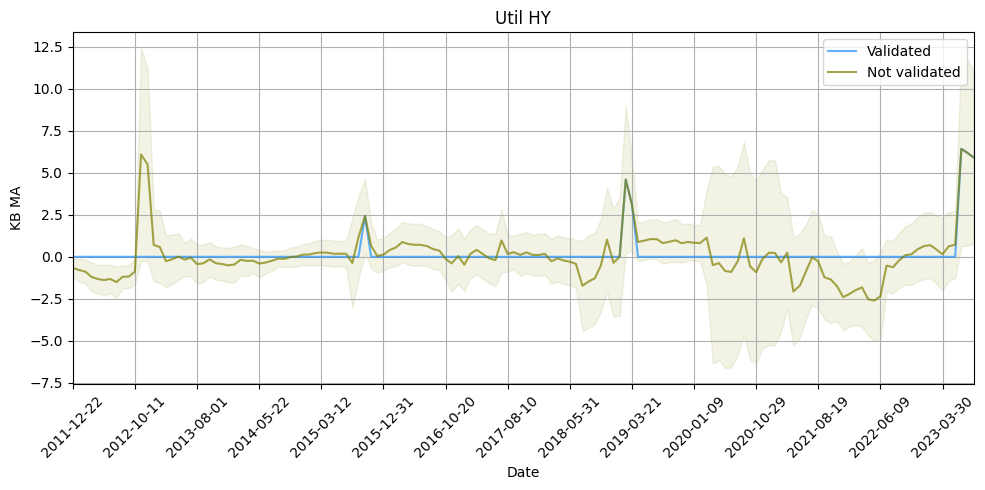}
  \caption{\footnotesize{Utilities HY}}
  \label{fig:2}
\end{subfigure}
\caption{\footnotesize{Validated, not validated and shaded confidence interval of CDX KB.}}
\label{fig:cdx_sector}
\end{figure}
\restoregeometry
\section{Details of the Data Set for \myref{Section}{sec:tick_data}}\label{sec:tick_data_table}
\footnotesize
\begin{table}[H]
\footnotesize
\centering
\caption{\footnotesize{Financial sector companies}}
\begin{tabular}{cccc}
\hline\hline
Ticker & Short name & Industry Group\tablefootnote{Based on Bloomberg Industrial Classification Standard (BICS) level 2.} \\\hline
AXP & American Express Co & Financial Services \\
AIG & American International Group & Insurance \\
BRK/B & Berkshire Hathaway Inc-CL B & Insurance \\
JPM & JPMorgan Chase \& Co & Banking \\
CNA & CNA Financial Corp & Insurance \\
LNC & Lincoln National Corp & Insurance \\
MMC & Marsh \& Mclennan Cos & Insurance \\
MBI & Mbia Inc & Insurance \\
MTG & Mgic Investment Corp & Financial Services \\
BAC & Bank of America Corp & Banking \\
WFC & Wells Fargo \& Co & Banking \\
C & Citigroup Inc & Banking \\
UNM & Unum Group & Insurance \\
MS & Morgan Stanley & Financial Services \\
ALL & Allstate Corp & Insurance \\
PRU & Prudential Financial Inc & Insurance \\
COF & Capital One Financial Corp & Financial Services \\
HIG & Hartford Financial SVCS GRP & Insurance \\
GS & Goldman Sachs Group Inc & Financial Services \\
MET & Metlife Inc & Insurance \\ \hline
\end{tabular}
\label{tab:cds_data}
\end{table}
\end{document}

%% file: schematic_network.tex
\begin{figure}[!tb]
\centering
\tikzset{every picture/.style={line width=0.75pt}} 
\begin{tikzpicture}[x=0.75pt,y=0.75pt,yscale=-1,xscale=1]
\draw   (67,126) .. controls (67,116.61) and (74.61,109) .. (84,109) .. controls (93.39,109) and (101,116.61) .. (101,126) .. controls (101,135.39) and (93.39,143) .. (84,143) .. controls (74.61,143) and (67,135.39) .. (67,126) -- cycle ;
\draw   (131,186) .. controls (131,176.61) and (138.61,169) .. (148,169) .. controls (157.39,169) and (165,176.61) .. (165,186) .. controls (165,195.39) and (157.39,203) .. (148,203) .. controls (138.61,203) and (131,195.39) .. (131,186) -- cycle ;
\draw   (130,125) .. controls (130,115.61) and (137.61,108) .. (147,108) .. controls (156.39,108) and (164,115.61) .. (164,125) .. controls (164,134.39) and (156.39,142) .. (147,142) .. controls (137.61,142) and (130,134.39) .. (130,125) -- cycle ;
\draw   (67,186) .. controls (67,176.61) and (74.61,169) .. (84,169) .. controls (93.39,169) and (101,176.61) .. (101,186) .. controls (101,195.39) and (93.39,203) .. (84,203) .. controls (74.61,203) and (67,195.39) .. (67,186) -- cycle ;
\draw    (84,143) -- (84,167) ;
\draw [shift={(84,169)}, rotate = 270] [color={rgb, 255:red, 0; green, 0; blue, 0 }  ][line width=0.75]    (10.93,-3.29) .. controls (6.95,-1.4) and (3.31,-0.3) .. (0,0) .. controls (3.31,0.3) and (6.95,1.4) .. (10.93,3.29)   ;
\draw    (101,126) -- (128,125.07) ;
\draw [shift={(130,125)}, rotate = 178.03] [color={rgb, 255:red, 0; green, 0; blue, 0 }  ][line width=0.75]    (10.93,-3.29) .. controls (6.95,-1.4) and (3.31,-0.3) .. (0,0) .. controls (3.31,0.3) and (6.95,1.4) .. (10.93,3.29)   ;
\draw    (148,143) -- (148,167) ;
\draw [shift={(148,169)}, rotate = 270] [color={rgb, 255:red, 0; green, 0; blue, 0 }  ][line width=0.75]    (10.93,-3.29) .. controls (6.95,-1.4) and (3.31,-0.3) .. (0,0) .. controls (3.31,0.3) and (6.95,1.4) .. (10.93,3.29)   ;
\draw    (101,186) -- (129,186) ;
\draw [shift={(131,186)}, rotate = 180] [color={rgb, 255:red, 0; green, 0; blue, 0 }  ][line width=0.75]    (10.93,-3.29) .. controls (6.95,-1.4) and (3.31,-0.3) .. (0,0) .. controls (3.31,0.3) and (6.95,1.4) .. (10.93,3.29)   ;
\draw    (135,176) -- (98.45,138.38) ;
\draw [shift={(97,137)}, rotate = 43.45] [color={rgb, 255:red, 0; green, 0; blue, 0 }  ][line width=0.75]    (10.93,-3.29) .. controls (6.95,-1.4) and (3.31,-0.3) .. (0,0) .. controls (3.31,0.3) and (6.95,1.4) .. (10.93,3.29)   ;
\draw    (147,108) .. controls (136,97) and (128,94) .. (117,94) .. controls (106.39,94) and (97.63,96.79) .. (85.35,107.77) ;
\draw [shift={(84,109)}, rotate = 317.29] [color={rgb, 255:red, 0; green, 0; blue, 0 }  ][line width=0.75]    (10.93,-3.29) .. controls (6.95,-1.4) and (3.31,-0.3) .. (0,0) .. controls (3.31,0.3) and (6.95,1.4) .. (10.93,3.29)   ;
\draw  [fill={rgb, 255:red, 245; green, 166; blue, 35 }  ,fill opacity=1 ] (279,100.5) .. controls (279,94.7) and (283.7,90) .. (289.5,90) .. controls (295.3,90) and (300,94.7) .. (300,100.5) .. controls (300,106.3) and (295.3,111) .. (289.5,111) .. controls (283.7,111) and (279,106.3) .. (279,100.5) -- cycle ;
\draw   (280,173.5) .. controls (280,167.7) and (284.7,163) .. (290.5,163) .. controls (296.3,163) and (301,167.7) .. (301,173.5) .. controls (301,179.3) and (296.3,184) .. (290.5,184) .. controls (284.7,184) and (280,179.3) .. (280,173.5) -- cycle ;
\draw   (279,135.5) .. controls (279,129.7) and (283.7,125) .. (289.5,125) .. controls (295.3,125) and (300,129.7) .. (300,135.5) .. controls (300,141.3) and (295.3,146) .. (289.5,146) .. controls (283.7,146) and (279,141.3) .. (279,135.5) -- cycle ;
\draw   (280,209.5) .. controls (280,203.7) and (284.7,199) .. (290.5,199) .. controls (296.3,199) and (301,203.7) .. (301,209.5) .. controls (301,215.3) and (296.3,220) .. (290.5,220) .. controls (284.7,220) and (280,215.3) .. (280,209.5) -- cycle ;
\draw   (356,101.5) .. controls (356,95.7) and (360.7,91) .. (366.5,91) .. controls (372.3,91) and (377,95.7) .. (377,101.5) .. controls (377,107.3) and (372.3,112) .. (366.5,112) .. controls (360.7,112) and (356,107.3) .. (356,101.5) -- cycle ;
\draw  [fill={rgb, 255:red, 245; green, 166; blue, 35 }  ,fill opacity=1 ] (356,136.5) .. controls (356,130.7) and (360.7,126) .. (366.5,126) .. controls (372.3,126) and (377,130.7) .. (377,136.5) .. controls (377,142.3) and (372.3,147) .. (366.5,147) .. controls (360.7,147) and (356,142.3) .. (356,136.5) -- cycle ;
\draw  [fill={rgb, 255:red, 245; green, 166; blue, 35 }  ,fill opacity=1 ] (356,174.5) .. controls (356,168.7) and (360.7,164) .. (366.5,164) .. controls (372.3,164) and (377,168.7) .. (377,174.5) .. controls (377,180.3) and (372.3,185) .. (366.5,185) .. controls (360.7,185) and (356,180.3) .. (356,174.5) -- cycle ;
\draw   (356,209.5) .. controls (356,203.7) and (360.7,199) .. (366.5,199) .. controls (372.3,199) and (377,203.7) .. (377,209.5) .. controls (377,215.3) and (372.3,220) .. (366.5,220) .. controls (360.7,220) and (356,215.3) .. (356,209.5) -- cycle ;
\draw  [fill={rgb, 255:red, 245; green, 166; blue, 35 }  ,fill opacity=1 ] (431,102.5) .. controls (431,96.7) and (435.7,92) .. (441.5,92) .. controls (447.3,92) and (452,96.7) .. (452,102.5) .. controls (452,108.3) and (447.3,113) .. (441.5,113) .. controls (435.7,113) and (431,108.3) .. (431,102.5) -- cycle ;
\draw   (432,137.5) .. controls (432,131.7) and (436.7,127) .. (442.5,127) .. controls (448.3,127) and (453,131.7) .. (453,137.5) .. controls (453,143.3) and (448.3,148) .. (442.5,148) .. controls (436.7,148) and (432,143.3) .. (432,137.5) -- cycle ;
\draw   (432,174.5) .. controls (432,168.7) and (436.7,164) .. (442.5,164) .. controls (448.3,164) and (453,168.7) .. (453,174.5) .. controls (453,180.3) and (448.3,185) .. (442.5,185) .. controls (436.7,185) and (432,180.3) .. (432,174.5) -- cycle ;
\draw  [fill={rgb, 255:red, 245; green, 166; blue, 35 }  ,fill opacity=1 ] (432,208.5) .. controls (432,202.7) and (436.7,198) .. (442.5,198) .. controls (448.3,198) and (453,202.7) .. (453,208.5) .. controls (453,214.3) and (448.3,219) .. (442.5,219) .. controls (436.7,219) and (432,214.3) .. (432,208.5) -- cycle ;
\draw  [fill={rgb, 255:red, 245; green, 166; blue, 35 }  ,fill opacity=1 ] (507,101.5) .. controls (507,95.7) and (511.7,91) .. (517.5,91) .. controls (523.3,91) and (528,95.7) .. (528,101.5) .. controls (528,107.3) and (523.3,112) .. (517.5,112) .. controls (511.7,112) and (507,107.3) .. (507,101.5) -- cycle ;
\draw  [fill={rgb, 255:red, 245; green, 166; blue, 35 }  ,fill opacity=1 ] (507,136.5) .. controls (507,130.7) and (511.7,126) .. (517.5,126) .. controls (523.3,126) and (528,130.7) .. (528,136.5) .. controls (528,142.3) and (523.3,147) .. (517.5,147) .. controls (511.7,147) and (507,142.3) .. (507,136.5) -- cycle ;
\draw  [fill={rgb, 255:red, 245; green, 166; blue, 35 }  ,fill opacity=1 ] (507,174.5) .. controls (507,168.7) and (511.7,164) .. (517.5,164) .. controls (523.3,164) and (528,168.7) .. (528,174.5) .. controls (528,180.3) and (523.3,185) .. (517.5,185) .. controls (511.7,185) and (507,180.3) .. (507,174.5) -- cycle ;
\draw   (507,209.5) .. controls (507,203.7) and (511.7,199) .. (517.5,199) .. controls (523.3,199) and (528,203.7) .. (528,209.5) .. controls (528,215.3) and (523.3,220) .. (517.5,220) .. controls (511.7,220) and (507,215.3) .. (507,209.5) -- cycle ;
\draw    (245,101) -- (277,100.53) ;
\draw [shift={(279,100.5)}, rotate = 179.16] [color={rgb, 255:red, 0; green, 0; blue, 0 }  ][line width=0.75]    (10.93,-3.29) .. controls (6.95,-1.4) and (3.31,-0.3) .. (0,0) .. controls (3.31,0.3) and (6.95,1.4) .. (10.93,3.29)   ;
\draw    (300,100.5) -- (354.32,135.42) ;
\draw [shift={(356,136.5)}, rotate = 212.74] [color={rgb, 255:red, 0; green, 0; blue, 0 }  ][line width=0.75]    (10.93,-3.29) .. controls (6.95,-1.4) and (3.31,-0.3) .. (0,0) .. controls (3.31,0.3) and (6.95,1.4) .. (10.93,3.29)   ;
\draw    (300,100.5) -- (354.79,172.91) ;
\draw [shift={(356,174.5)}, rotate = 232.88] [color={rgb, 255:red, 0; green, 0; blue, 0 }  ][line width=0.75]    (10.93,-3.29) .. controls (6.95,-1.4) and (3.31,-0.3) .. (0,0) .. controls (3.31,0.3) and (6.95,1.4) .. (10.93,3.29)   ;
\draw    (377,136.5) -- (429.31,103.57) ;
\draw [shift={(431,102.5)}, rotate = 147.8] [color={rgb, 255:red, 0; green, 0; blue, 0 }  ][line width=0.75]    (10.93,-3.29) .. controls (6.95,-1.4) and (3.31,-0.3) .. (0,0) .. controls (3.31,0.3) and (6.95,1.4) .. (10.93,3.29)   ;
\draw    (377,136.5) -- (430.79,206.91) ;
\draw [shift={(432,208.5)}, rotate = 232.62] [color={rgb, 255:red, 0; green, 0; blue, 0 }  ][line width=0.75]    (10.93,-3.29) .. controls (6.95,-1.4) and (3.31,-0.3) .. (0,0) .. controls (3.31,0.3) and (6.95,1.4) .. (10.93,3.29)   ;
\draw    (377,174.5) -- (430.3,207.45) ;
\draw [shift={(432,208.5)}, rotate = 211.72] [color={rgb, 255:red, 0; green, 0; blue, 0 }  ][line width=0.75]    (10.93,-3.29) .. controls (6.95,-1.4) and (3.31,-0.3) .. (0,0) .. controls (3.31,0.3) and (6.95,1.4) .. (10.93,3.29)   ;
\draw    (453,208.5) -- (506.1,103.29) ;
\draw [shift={(507,101.5)}, rotate = 116.78] [color={rgb, 255:red, 0; green, 0; blue, 0 }  ][line width=0.75]    (10.93,-3.29) .. controls (6.95,-1.4) and (3.31,-0.3) .. (0,0) .. controls (3.31,0.3) and (6.95,1.4) .. (10.93,3.29)   ;
\draw    (452,102.5) -- (506.32,137.42) ;
\draw [shift={(508,138.5)}, rotate = 212.74] [color={rgb, 255:red, 0; green, 0; blue, 0 }  ][line width=0.75]    (10.93,-3.29) .. controls (6.95,-1.4) and (3.31,-0.3) .. (0,0) .. controls (3.31,0.3) and (6.95,1.4) .. (10.93,3.29)   ;
\draw    (452,102.5) -- (506.79,174.91) ;
\draw [shift={(508,176.5)}, rotate = 232.88] [color={rgb, 255:red, 0; green, 0; blue, 0 }  ][line width=0.75]    (10.93,-3.29) .. controls (6.95,-1.4) and (3.31,-0.3) .. (0,0) .. controls (3.31,0.3) and (6.95,1.4) .. (10.93,3.29)   ;

\draw (75,119) node [anchor=north west][inner sep=0.75pt]   [align=left] {W};
\draw (140,119) node [anchor=north west][inner sep=0.75pt]   [align=left] {X};
\draw (141,179) node [anchor=north west][inner sep=0.75pt]   [align=left] {Z};
\draw (76,179) node [anchor=north west][inner sep=0.75pt]   [align=left] {Y};
\draw (227,92) node [anchor=north west][inner sep=0.75pt]   [align=left] {W};
\draw (257,128) node [anchor=north west][inner sep=0.75pt]   [align=left] {X};
\draw (256,166) node [anchor=north west][inner sep=0.75pt]   [align=left] {Y};
\draw (257,202) node [anchor=north west][inner sep=0.75pt]   [align=left] {Z};
\draw (272,65) node [anchor=north west][inner sep=0.75pt]   [align=left] {{\small t = 0}};
\draw (352,65) node [anchor=north west][inner sep=0.75pt]   [align=left] {{\small t = 1}};
\draw (427,65) node [anchor=north west][inner sep=0.75pt]   [align=left] {{\small t = 2}};
\draw (502,66) node [anchor=north west][inner sep=0.75pt]   [align=left] {{\small t = 3}};
\draw (256,83.4) node [anchor=north west][inner sep=0.75pt]    {$\epsilon $};
\draw (556,90.4) node [anchor=north west][inner sep=0.75pt]    {$\cdots $};
\draw (556,127.4) node [anchor=north west][inner sep=0.75pt]    {$\cdots $};
\draw (556,201.4) node [anchor=north west][inner sep=0.75pt]    {$\cdots $};
\draw (556,166.4) node [anchor=north west][inner sep=0.75pt]    {$\cdots $};
\end{tikzpicture}
\caption{\footnotesize{\textbf{Left} schematic network \textbf{Right} shock propagation through network over time}}
\label{fig:netw_schem}
\end{figure}